\documentclass[a4paper,11pt]{article}
\pdfoutput=1 

\usepackage{jheppub} 

\usepackage[T1]{fontenc} 
\usepackage{subcaption}
\usepackage{verbatim}
\usepackage{mathrsfs}
\usepackage{physics}
\usepackage{amsthm}
\usepackage{mathtools}
\usepackage{hyperref}
\usepackage{enumerate}

\renewcommand{\tilde}{\widetilde}
\renewcommand{\bar}{\overline}

\newcommand{\comps}{\mathbb{C}}
\renewcommand{\Im}{\text{Im}}
\renewcommand{\Re}{\text{Re}}

\newcommand{\A}{\mathcal{A}}

\renewcommand{\H}{\mathcal{H}}
\newcommand{\K}{\mathcal{K}}

\newcommand{\B}{\mathcal{B}}
\newcommand{\D}{\mathcal{D}}

\newtheorem{theorem}{Theorem}
\newtheorem{corollary}[theorem]{Corollary}
\newtheorem{lemma}[theorem]{Lemma}
\newtheorem{prop}[theorem]{Proposition}
\newtheorem{definition}[theorem]{Definition}

\renewcommand{\op}[1]{#1}

\theoremstyle{definition}
\newtheorem{remark}[theorem]{Remark}

\numberwithin{theorem}{section}

\title{An intuitive construction of modular flow}
\author{Jonathan Sorce}
\affiliation{MIT Center for Theoretical Physics}
\abstract{
The theory of modular flow has proved extremely useful for applying thermodynamic reasoning to out-of-equilibrium states in quantum field theory.
However, the standard proofs of the fundamental theorems of modular flow use machinery from Fourier analysis on Banach spaces, and as such are not especially transparent to an audience of physicists.
In this article, I present a construction of modular flow that differs from existing treatments.
The main pedagogical contribution is that I start with thermal physics via the KMS condition, and derive the modular operator as the only operator that could generate a thermal time-evolution map, rather than starting with the modular operator as the fundamental object of the theory.
The main technical contribution is a new proof of the fundamental theorem stating that modular flow is a symmetry.
The new proof circumvents the delicate issues of Fourier analysis that appear in previous treatments, but is still mathematically rigorous.
}

\makeatletter
\gdef\@fpheader{MIT preprint ID MIT-CTP/5622\vspace{3em}}
\makeatother
\begin{document}
\maketitle

\section{Introduction}

Given a Hamiltonian $H$, the density matrix
\begin{equation}
	\rho_{\beta} = \frac{e^{- \beta H}}{\tr(e^{-\beta H})}
\end{equation}
is said to be thermal with respect to $H$ at inverse temperature $\beta$.
Conversely, given an invertible density matrix $\rho,$ we can always construct a Hamiltonian with respect to which it is thermal, given by
\begin{equation}
	K_{\rho} \equiv - \log\rho.
\end{equation}
The operator $K_{\rho}$ is self-adjoint and bounded below, and is called the \textit{modular Hamiltonian} of the state $\rho.$
While $K_{\rho}$ is generally a highly nonlocal operator, it can be conceptually helpful to think of it as a physical Hamiltonian, since information-theoretic quantities associated to $\rho$ can be treated as thermodynamic quantities associated to $K_{\rho}.$

Many quantum systems of interest do not admit density matrices.
Typical examples are thermodynamic systems at infinite volume \cite{araki-woods, HHW, witten-limits} or local subregions in quantum field theory \cite{araki-fields, fredenhagen1985modular}.
Nevertheless, in these settings it is still possible to construct an operator that plays the same role played by $K_{\rho}$ in the finite-dimensional setting.
The theory of this operator is known as \textit{Tomita-Takesaki theory} or \textit{modular theory}.
It was first developed in \cite{takesaki2006tomita}, is treated in the textbooks \cite{takesaki-book, struatilua2019lectures, sunder-book}, and has been reviewed for physicists in \cite{witten-notes}.
In recent years it has reemerged in high energy physics as a tool for studying energy and entropy in quantum field theory and in semiclassical gravity; for an incomplete list of examples, see \cite{Casini:2011kv, Wall:2011hj, Bousso:2014sda, Faulkner:2016mzt, Cardy:2016fqc, Hollands:2017dov, Casini:2017roe, Faulkner:2017vdd, Balakrishnan:2017bjg, Chen:2018rgz,Faulkner:2018faa, Lashkari:2018nsl, Jefferson:2018ksk, Ceyhan:2018zfg, Lashkari:2019ixo, Chen:2019iro, Faulkner:2020iou, Leutheusser:2021frk, Witten:2021unn, Chandrasekaran:2022cip, Chandrasekaran:2022eqq, Leutheusser:2022bgi, Penington:2023dql, Parrikar:2023lfr, Furuya:2023fei, Jensen:2023yxy, AliAhmad:2023etg, Klinger:2023tgi, Kudler-Flam:2023qfl}.

The setting for Tomita-Takesaki theory is a Hilbert space $\H$ with a von Neumann algebra $\A$ describing the quantum degrees of freedom in some subsystem.
The traditional development of the theory begins with a state $\ket{\Omega} \in \H$ that is ``cyclic and separating'' for $\A$ (defined in section \ref{sec:cyclic-separating}), and defines an antilinear operator $S_{\Omega}$ that acts as
\begin{equation}
	S_{\Omega}(\op{a} \ket{\Omega}) = \op{a}^{\dagger} \ket{\Omega}, \quad \op{a} \in \A.
\end{equation}
This is an unbounded operator, but it is sufficiently well behaved that the operator $\Delta_{\Omega} \equiv S_{\Omega}^{\dagger} S_{\Omega}$ is densely defined, positive, and invertible.
This is called the \textit{modular operator}, and can be used to define the \textit{modular Hamiltonian} $K_{\Omega} \equiv - \log \Delta_{\Omega}.$
The modular Hamiltonian is a self-adjoint operator, so it generates a unitary group that can be used to act on operators in $\A$ via the map
\begin{equation}
	\op{a} \mapsto e^{i K_{\Omega} t} \op{a} e^{-i K_{\Omega} t} = \Delta_{\Omega}^{-it} \op{a} \Delta_{\Omega}^{it}.
\end{equation}
This is called the modular flow of operators.

The next step in the traditional development of the theory is to show that modular flow maps $\A$ into itself.
This statement, known as Tomita's theorem, essentially says that the modular flow of a quantum subsystem does not mix the subsystem with its complement.
There are many approaches to proving Tomita's theorem \cite{takesaki2006tomita, van-daele-proof, zsido-proof, rieffel-proof, longo-proof, woronowicz-proof}, but the general method is to take an integral transform of the function $f(t) = \Delta_{\Omega}^{-it} \op{a} \Delta_{\Omega}^{it},$ show that the transformed function lies in $\A,$ and then to show that the inverse integral transform converges in a topology for which $\A$ is closed.
After proving Tomita's theorem, one shows that the modular flow of operators satisfies something called the \textit{KMS condition}, which basically means that the state $\ket{\Omega}$, restricted to the subsystem described by the algebra $\A$, looks thermal with respect to the time-evolution map generated by the Hamiltonian $K_{\Omega}.$

In this article, I present an approach to Tomita-Takesaki theory that differs from the one described in the preceding paragraphs.
Pedagogically, the main difference is that I do not start with the antilinear operator $S_{\Omega},$ whose introduction I have always found fairly ad hoc.
Instead, I start with the KMS condition, which characterizes thermality in terms of time-evolved two-point functions.
I prove a uniqueness theorem, showing that if there exists \textit{some} Hermitian operator $H$ for which $\ket{\Omega}$ satisfies the KMS condition in subsystem $\A$, then $H$ must be the modular Hamiltonian $K_{\Omega}$.\footnote{This result is known to mathematicians, and appears e.g. in the textbooks \cite{takesaki-book, struatilua2019lectures}.}
In this way of developing the theory, thermal physics comes first: the modular Hamiltonian $K_{\Omega}$ arises not as an abstraction, but as the only viable candidate for a thermal time-evolution map.
This gives a way of understanding why modular flow is so useful for studying quantum field theory and semiclassical gravity: if you want to apply thermodynamic reasoning to general states, then modular flow is the only tool available.

I then give a new proof of Tomita's theorem which differs from existing proofs, and which I hope is easier to understand.
The proof works by showing that for any $\op{a}$ in $\A$ and any $\op{b}'$ in the commutant algebra $\A'$, the commutator $[\Delta_{\Omega}^{-it} \op{a} \Delta_{\Omega}^{it}, \op{b}']$ vanishes.
By the bicommutant theorem of von Neumann, this implies that $\Delta_{\Omega}^{-it} a \Delta_{\Omega}^{it}$ must be in $\A.$
To show that the commutator vanishes, I construct a dense subset of operators $a$ in $\A$, called the ``tidy subspace'' of $\A,$ for which the map $it \mapsto [\Delta_{\Omega}^{-it} a \Delta_{\Omega}^{it}, b']$ admits an analytic extension to the full complex plane, such that the norm of this analytic extension is bounded at infinity by an exponential function.
I show that the analytic extension vanishes on the integers, and then apply a version of Carlson's theorem to show that the commutator is identically zero.
This proves Tomita's theorem for the tidy subspace, and the general case is obtained by a continuity argument.

Since the goal of this paper is to make a complicated mathematical theory palatable and useful for theoretical physicists, I have included pedagogical treatments of background material in a review section and in several appendices.
The outline of the paper is given in a bulleted list below.
Minimally, I suggest reading section \ref{sec:uniqueness} and the preamble to section \ref{sec:tomitas-theorem}, and consulting section \ref{sec:background} on an as-needed basis to fill in mathematical background.
This will give you the main ideas of the perspective I am taking on modular flow and on Tomita's theorem.
The rest of the paper is more technical, and is aimed at readers who wish to work directly with modular flow themselves.
\begin{itemize}
	\item In section \ref{sec:background}, I give mathematical background that is needed for the rest of the paper.
	I discuss several useful topologies on spaces of operators; I explain basic properties of the ``cyclic and separating'' states that are the continuum generalization of bipartite states with full Schmidt rank; I state some basic theorems concerning unbounded operators; and I explain the theory of contour integration and analytic continuation for operator-valued functions.
	\item In section \ref{sec:uniqueness}, I introduce the KMS condition as a characterization of thermality in terms of two-point functions.
	I prove the KMS uniqueness theorem, introducing the modular operator by showing that for a cyclic and separating state, the modular Hamiltonian is the only Hamiltonian that could possibly make that state look thermal in the sense of KMS.
	I also observe that the modular Hamiltonian will automatically satisfy the KMS condition if Tomita's theorem holds.
	\item In section \ref{sec:tomitas-theorem}, I prove Tomita's theorem.
	I also prove the related statement that the modular conjugation operator maps $\A$ to $\A'$.
	\item In three appendices, I explain Carlson's theorem, prove the basic properties of the Tomita operator, and outline the textbook proof of Tomita's theorem.
\end{itemize}
A shorter version of my proof of Tomita's theorem is presented in the companion paper \cite{sorce-short-proof}, aimed at an audience of mathematicians.

\subsection{Historical remarks}

Four general proofs of Tomita's theorem have been given previously, by Takesaki \cite{takesaki2006tomita}, van Daele \cite{van-daele-proof}, Zsid\'{o} \cite{zsido-proof}, and Woronowicz \cite{woronowicz-proof}.
Of these, van Daele's proof is the one that appears in textbooks (see e.g. \cite{takesaki-book, struatilua2019lectures}).
There is also a proof of Tomita's theorem due to Longo \cite{longo-proof} in the case where the von Neumann algebra $\A$ is hyperfinite.
This property is believed to hold for the von Neumann algebras that appear in quantum field theory \cite{buchholz1987universal}.

The proofs by Takesaki, van Daele, and Zsid\'{o} all make use of a lemma due to Takesaki concerning the resolvent of the modular operator.
This lemma is also used in my proof, and is the subject of section \ref{sec:resolvent-lemma}.
Immediately after introducing that lemma, van Daele's proof diverges from mine, and proceeds by using the resolvent lemma to study Fourier transforms of modular flow; the general outlines of this proof are explained in appendix \ref{app:van-daele}.
Zsid\'{o}'s proof, by contrast, uses Takesaki's resolvent lemma to construct a dense subset of $\A$ for which modular flow admits an analytic continuation to the full complex plane.
I undertake a similar construction, but my dense subset is different from Zsid\'{o}'s, and has the additional property that the analytic continuation of modular flow is bounded by an exponential function at infinity.
This allows me to use Carlson's theorem to prove Tomita's theorem by studying the integer behavior of the analytic extension.
By contrast, Zsid\'{o}'s proof expresses analytic continuations of modular flow in terms of the theory of analytic generators \cite{cioranescu1976analytic}, and appeals to results from that theory (essentially using Mellin transforms in place of van Daele's Fourier transforms) to finish the proof of Tomita's theorem.

The technique of applying Carlson's theorem to analytic extensions of modular flow was inspired by a proof of Tomita's theorem given by Bratteli and Robinson \cite[pages 90-91]{bratteli2012operator} in the special case of a bounded modular operator.

\section{Mathematical background}
\label{sec:background}

This is the most technical section of the paper, and unfortunately that is unavoidable, as one cannot understand the statement and proof of Tomita's theorem without first understanding the basics of operator theory.
I have tried to write this section in a way that emphasizes the ``rules of the game'' for manipulating unbounded operators, without getting bogged down in specific technical lemmas.
Eager readers may wish to proceed immediately to section \ref{sec:uniqueness} and consult this section on an as-needed basis.

Much of the material presented in this section can be found in the textbooks \cite{rudin1991functional, conway2000course, struatilua2019lectures}.
Some of it is also discussed in my recent review article \cite{sorce-review}.

\subsection{Operator topologies}
\label{subsec:operator-topologies}

Let $\H$ be a Hilbert space, and let $\B(\H)$ denote the space of bounded operators on $\H$.
There are many interesting topologies on $\B(\H).$
Among these, five appear most commonly: norm, ultrastrong, ultraweak, strong, and weak.
The only topologies that appear in this paper are the norm, strong, and weak topologies.
However, since the ultraweak and ultrastrong topologies are useful for studying operator algebras, and since several people have told me they find these topologies confusing, I have decided to include them in this section.
Figure \ref{fig:topologies} shows a flowchart of these topologies in order of strength.
This figure should be understood to mean, for example, that any sequence convergent in the norm topology is convergent in the ultrastrong topology, but there may be sequences that converge in the ultrastrong topology which do not converge in the norm topology.

\begin{figure}
	\centering
	\includegraphics{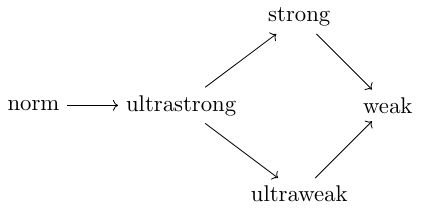}
	\caption{A flowchart of the standard topologies on operator algebras.
		Arrows point from stronger topologies to weaker topologies.
		So, for example, every sequence that converges in norm also converges ultrastrongly.}
	\label{fig:topologies}
\end{figure}

The norm topology is the easiest topology to understand.
A bounded operator $T : \H \to \H$ has, by definition, some constant $c$ satisfying
\begin{equation} \label{eq:norm-bounded}
	\lVert T \ket{\psi} \rVert \leq c \lVert \ket{\psi} \rVert, \quad \ket{\psi} \in \H.
\end{equation}
The operator norm of $T$ is defined as the infimum over all constants satisfying this inequality; equivalently, as
\begin{equation}
	\lVert T \rVert_{\infty} = \sup_{\ket{\psi} \text{ s.t. } \lVert \ket{\psi} \rVert \leq 1} \lVert T \ket{\psi} \rVert.
\end{equation}
The norm topology on $\B(\H)$ is defined so that a net\footnote{A net is a generalization of a sequence, and in general topological spaces, convergence of sequences is not enough to completely determine the topology.
For correctness I will always use the word ``net'' when necessary, but it is conceptually fine to think in terms of sequences.} $T_{\alpha} \in \B(\H)$ converges to $T \in \B(\H)$ if and only if we have
\begin{equation}
	\lVert T_{\alpha} - T \rVert_{\infty} \to 0.
\end{equation}
The norm topology is best thought of as a topology of uniform convergence: $T_{\alpha}$ must converge to $T$ uniformly on all vectors.

The remaining four topologies are ultrastrong, ultraweak, strong, and weak.
A useful mnemonic for understanding these topologies is as follows.
\begin{quote}
	A ``strong'' topology is one where an operator is said to be small if its action on any state is small.
	A ``weak'' topology is one where an operator is said to be small if its expectation value in any state is small.
	The modifier ``ultra'' indicates that one must check smallness on mixed states as well as pure states; in a topology without the ``ultra'' modifier, smallness is determined only with respect to pure states.
\end{quote}
Concretely, we have:
\begin{itemize}
	\item Strong: $T_{\alpha} \to T$ in the strong topology if, for every $\ket{\psi} \in \H,$ we have $(T - T_{\alpha}) \ket{\psi} \to 0.$
	\item Weak: $T_{\alpha} \to T$ in the weak topology if, for every $\ket{\psi} \in \H,$ we have $\bra{\psi} (T - T_{\alpha}) \ket{\psi} \to 0.$
	\item Ultrastrong: $T_{\alpha} \to T$ in the ultrastrong topology if, for every positive trace-class operator $\rho \in \H,$ we have $(T - T_{\alpha}) \rho (T -T_{\alpha})^{\dagger} \to 0$ in the trace norm.\footnote{Note that this is different from the condition $T_{\alpha} \rho T_{\alpha}^{\dagger} \to T \rho T^{\dagger}.$
	That condition, since it is not a function of the difference $T - T_{\alpha}$, is not compatible with the vector space structure of $\B(\H).$}
	\item Ultraweak: $T_{\alpha} \to T$ in the ultraweak topology if, for every positive trace-class operator $\rho \in \H,$ we have $\tr(\rho (T - T_{\alpha})) \to 0.$
\end{itemize}

A subalgebra $\A$ of $\B(\H)$ that is closed under adjoints and that contains the identity is called a (unital) $*$-algebra.
If $\A$ is closed in the norm topology, it is called a C* algebra.
If it is closed in the weak topology, it is called a von Neumann algebra.
From figure \ref{fig:topologies}, we can see that the closure of a set in the norm topology is generally smaller than its closure in any other topology; so every von Neumann algebra is a C* algebra, but not every C* algebra is a von Neumann algebra.

Von Neumann algebras are the objects that naturally arise to describe the set of operators associated with a subregion of a physical system; see \cite[section 2.6]{witten-notes} for more discussion on this point.
When studying von Neumann algebras, the most important theorem is von Neumann's double commutant theorem \cite{v1930algebra}.
The commutant of $\A$, denoted $\A',$ is defined as the set of all operators that commute with everything in $\A.$
The double commutant theorem is
\begin{equation}
	\A'' = \A.
\end{equation}
This is an extremely useful theorem, because to check that an operator is in $\A,$ it suffices to check that it commutes with every operator in $\A'.$

\subsection{Cyclic and separating states}
\label{sec:cyclic-separating}

As an example of a von Neumann algebra describing a quantum subsystem, we may think of the Hilbert space $\H = \H_1 \otimes \H_2,$ and take $\A$ to be the set of operators acting only on the first tensor factor:
\begin{equation}
	\A = \B(\H_1) \otimes 1_{\H_2}.
\end{equation}
It is easy to check that this is a von Neumann algebra, and that its commutant is given by
\begin{equation}
	\A'
		= 1_{\H_1} \otimes \B(\H_2).
\end{equation}
Given any vector $\ket{\Omega} \in \H,$ one can produce a Schmidt decomposition:
\begin{equation}
	\ket{\Omega}
		= \sum_{j} \sqrt{p_j} \ket{\psi_j}_{\H_1} \otimes \ket{\phi_j}_{\H_2}.
\end{equation}
If the Schmidt eigenvectors form complete bases of $\H_1$ and of $\H_2$, then we say the state $\ket{\Omega}$ is ``fully entangled'' or has ``full Schmidt rank;'' in this case, the reduced density matrices of $\ket{\Omega}$ on $\H_1$ and $\H_2$ are both invertible.

For a general von Neumann algebra $\A,$ there is generally no tensor product structure for which $\A$ corresponds to a tensor factor, and one cannot talk about the Schmidt rank of a vector $\ket{\Omega}.$
There is, however, a more general criterion that captures the essential physics of a state having full Schmidt rank: this is the condition that $\ket{\Omega}$ is \textit{cyclic and separating}.

Formally, the state $\ket{\Omega}$ is said to be cyclic for $\A$ if the set
\begin{equation}
	\A \ket{\Omega} \equiv \{a \ket{\Omega}\}_{a \in \A}
\end{equation}
is dense in $\H$.
It is said to be separating for $\A$ if it is cyclic for $\A',$ i.e., if the set $\A' \ket{\Omega}$ is dense in $\H$.
It is a straightforward exercise to check that this property is satisfied by any full-Schmidt-rank state in a tensor product decomposition.
A useful equivalent condition, which is not so hard to prove, is as follows.
\begin{prop}
	The state $\ket{\Omega}$ is separating for $\A$ if and only if it distinguishes operators on $\A$ in the sense that if an operator $a \in \A$ satisfies $a \ket{\Omega} = 0$, then we must have $a=0.$
\end{prop}
\noindent This proposition explains the terminology ``separating'' --- the state $\ket{\Omega}$ ``separates'' elements of $\A$ in that it can be used to distinguish them perfectly.

The Reeh-Schlieder theorem \cite{reeh1961bemerkungen, strohmaier2000reeh, strohmaier2002microlocal} guarantees that in quantum field theory, states with well behaved short-distance singularities are cyclic and separating.
See \cite[section 2]{witten-notes} for a review of this point.

\subsection{Unbounded operators}
\label{sec:unbounded-operators}

Tomita-Takesaki theory involves unbounded operators.
This subsection provides a short exposition of the essential features of the theory of such operators.

Given a Hilbert space $\H$ and a Hilbert space $\K$, an operator from $\H$ to $\K$ is defined to be a linear map $T$ from a subspace $\D_T \subseteq \H,$ called the \textit{domain} of $T$, into $\K$.
Note that this generalizes the standard definition of an operator by allowing $T$ to act only on some subspace of $\H$.
This generalization is essential when considering unbounded operators, which may not be well defined for arbitrary input vectors.

An operator $T$ is said to be \textit{bounded} if it is continuous in the norm topology; i.e., if whenever $\ket{\psi_n} \in \D_T$ converges to $\ket{\psi}$ as a sequence of vectors, we have $\ket{\psi} \in \D_T$ and $T \ket{\psi_n} \to T \ket{\psi}.$
This is equivalent to the statement that $T$ satisfies inequality \eqref{eq:norm-bounded}.
For unbounded operators, even if the sequence $\ket{\psi_n}$ converges, the sequence $T \ket{\psi_n}$ need not converge.
There is, however, a useful generalization of the boundedness condition which picks out a special class of unbounded operators.
\begin{definition}
	An operator $T$ from $\H$ to $\K$ is said to be \textbf{closed} if, whenever $\ket{\psi_n} \in \D_T$ converges to $\ket{\psi}$ and $T \ket{\psi_n}$ is a convergent sequence, then we have $\ket{\psi} \in \D_T$ and $T \ket{\psi_n} \to T \ket{\psi}.$ An operator $T$ from $\H$ to $\K$ is said to be \textbf{preclosed} or \textbf{closable} if it can be extended to a closed operator on a larger domain.
\end{definition}
\noindent One can show (see for example \cite[chapter 13]{rudin1991functional}) that a preclosed operator has a unique smallest closed extension, constructed by adding $\ket{\psi}$ to the domain of $T$ whenever $\ket{\psi}$ is the endpoint of a sequence $\ket{\psi_n}$ for which $T \ket{\psi_n}$ converges.
If $T$ is preclosed, then its smallest closed extension is denoted $\bar{T}.$
If the operator $S_2$ is an extension of the operator $S_1,$ we write $S_1 \subseteq S_2.$

It is useful to keep in mind that if we have a closed operator $T$ with domain $\D_T,$ then for any subspace $V \subseteq \D_T,$ the restriction $T|_V$ has a closed extension (namely, $T$), and is therefore preclosed.
If the subspace $V$ is too small, then it might happen that the closure of $T|_{V}$ is a closed operator on some proper subspace of $\D_T,$ so that we have $\bar{T|_{V}} \subsetneq T.$
By contrast, if $V$ is a subspace such that $\bar{T|_{V}}$ is equal to $T$, then $V$ is said to be a \textit{core} of $T.$
Given a core for $T$, the action of $T$ on any vector in its domain can be written as a limit of the action of $T$ on vectors in the core.
Consequently, to prove that a closed operator $T$ satisfies some property, it is often sufficient to check that this property is satisfied on a core.
This idea will be so important later on that we promote it to a formal definition for easy reference.
\begin{definition} \label{def:core}
	Given a closed operator $T$ with domain $\D_T,$ a \textbf{core} for $T$ is a subspace $V \subseteq \D_T$ such that the restriction of $T$ to $V$, which is a preclosed operator, has as its closure the full operator $T$.
	In an equation:
	\begin{equation}
		\bar{T|_V} = T.
	\end{equation}
\end{definition}
\begin{remark} \label{rem:core-graph}
	One useful way to think about closed operators is in terms of an object called a ``graph.''
	The graph of an operator $T$ from $\H$ to $\K$ is a vector subspace of the Hilbert space $\H \oplus \K$, and which is useful because all of the information about how $T$ acts is contained in the structure of this subspace.
	It is defined by
	\begin{equation}
		\text{Graph}(T) = \{\ket{x} \oplus T \ket{x}\}_{\ket{x} \in \D_T}.
	\end{equation}
	From what we have discussed so far, it is not so hard to see that $T$ is closed as an operator if and only if $\text{Graph}(T)$ is topologically closed as a subset of $\H \oplus \K.$
	The graph of the closure of an operator is equal to the topological closure of its graph.
	In particular, to show that a subspace $V$ is a core for the operator $T$, it suffices to show that the set
	\begin{equation}
		\text{Graph}(T|_V) = \{\ket{x} \oplus T \ket{x}\}_{\ket{x} \in V} 
	\end{equation}
	is dense in $\text{Graph}(T)$ with respect to the Hilbert space topology on $\H \oplus \K.$
	If this perspective does not feel helpful, feel free to put it aside for now; we will use it exactly once in the paper, in one of the final steps of the proof of Tomita's theorem in section \ref{sec:tidy-subspace}.
\end{remark}

Now we will discuss adjoints of unbounded operators.
The operator $T$ from $\H$ to $\K$ is said to be \textbf{densely defined} if the domain $\D_T$ is dense in $\H$.
For any densely defined operator, we can define an adjoint operator $T^{\dagger}$ from $\K$ to $\H.$
One must be a little careful in defining this, as $T^{\dagger}$ will generally only be defined on some subspace of the Hilbert space $\K$.
We would like $T^{\dagger}$ to satisfy the standard defining equation
\begin{equation}
	\langle \psi | T \xi \rangle_{\K} = \langle T^{\dagger} \psi | \xi \rangle_{\H}, \quad \ket{\xi} \in \D_T, \ket{\psi} \in \D_{T^{\dagger}}.
\end{equation}
To make sense of this equation, we should take the domain of $T^{\dagger}$ to be the set of all vectors $\ket{\psi}$ in $\K$ for which there exists a vector $|\psi' \rangle \in \H$ satisfying
\begin{equation}
	\braket{\psi}{T \xi}_{\K} = \langle \psi'| \xi \rangle_\H
\end{equation}
for every $\ket{\xi}$ in the domain of $T$.
The Riesz lemma \cite[theorem 2.5]{rudin1991functional} says that such a vector exists if and only if the map
\begin{equation} \label{eq:adjoint-domain}
	\ket{\xi} \mapsto \braket{\psi}{T \xi}_{\K}, \quad \ket{\xi} \in \D_T
\end{equation}
is bounded.
Density of $\D_T$ in $\H$ tells us that $\ket{\psi'}$ is uniquely determined.
If we take $\D_{T^{\dagger}}$ to be the space of all vectors $\ket{\psi}$ for which the map \eqref{eq:adjoint-domain} is bounded, then we can consistently define $T^{\dagger}$ on this domain as the map that takes $\ket{\psi}$ to $\ket{\psi'}.$

There is a subtlety we must now address, which is that even if $T$ is densely defined, its adjoint $T^{\dagger}$ may not be densely defined, so we cannot always take the adjoint twice.
This is addressed by the following proposition.
(See for example \cite[chapter 13]{rudin1991functional}.)
\begin{prop}[Properties of adjoints] \label{prop:adjoint-properties}
	A densely defined operator $T$ is preclosed if and only if its adjoint is densely defined.
	In this case, we have $T^{\dagger} = (\bar{T})^{\dagger},$ and $T^{\dagger \dagger} = \bar{T}.$
	Furthermore, $T^{\dagger}$ is always closed, and if $S$ is an extension of $T$, i.e. $T \subseteq S,$ then we have $S^{\dagger} \subseteq T^{\dagger}.$
\end{prop}

Now that we understand the properties of adjoints, we may define a \textit{Hermitian} or \textit{self-adjoint} operator as an operator from $\H$ to $\H$ satisfying $T = T^{\dagger}.$
Note that for this equation to hold, the domains of $T$ and $T^{\dagger}$ must be the same!
Note also that since adjoints are closed, any self-adjoint operator is automatically closed.
For self-adjoint operators, even if they are unbounded, we have the following version of the spectral theorem.
For a proof, see \cite[chapters 12 and 13]{rudin1991functional}.
\begin{definition}
	A subset of a topological space is \textbf{Borel} if it can be written using unions and intersections of at most countably many open or closed sets.
\end{definition}
\begin{definition}
	The \textbf{spectrum} of an operator $T$ is the set of all complex numbers $z$ for which $z-T$ cannot be inverted as a bounded operator.
\end{definition}
\begin{theorem}[Spectral theorem] \label{thm:spectral-theorem}
	Let $T : \D_T \to \H$ be a self-adjoint operator in $\H.$
	Then the spectrum of $T$ is a subset of the real line, and to every Borel subset $\omega$ of the spectrum, there is an associated spectral projection $P_{\omega},$ which is a projection operator that commutes with $T$.
	Projectors for disjoint subsets of the spectrum project onto orthogonal subspaces, and the projector for a countable union of disjoint sets is formed by taking the countable sum of the projectors for each set.
	The spectral projection for the empty set projects onto the zero vector; the spectral projection for the full spectrum is the identity operator.
	
	For any vectors $\ket{\psi}, \ket{\xi} \in \H,$ the function
	\begin{equation}
		\mu_{\psi, \xi}(\omega) = \langle \psi | P_{\omega} | \xi \rangle
	\end{equation}
	is a measure on the spectrum of $T$.
	From this, one may define an operator $f(T)$ for $f$ any measurable, complex function of the spectrum.
	This operator is closed and densely defined, and has the following properties.
	\begin{itemize}
		\item The domain of $f(T)$ is the set of all vectors $\ket{\psi}$ for which we have
		\begin{equation}
			\int_{\text{spectrum}} d\mu_{\psi, \psi} |f|^2 < \infty.
		\end{equation}
		This integral gives the norm-squared of the vector $f(T) \ket{\psi}.$
		\item Given $\ket{\psi}$ in the domain of $f(T)$ and $\ket{\xi}$ in $\H$, the matrix element is given by
		\begin{equation}
			\bra{\xi} f(T) \ket{\psi} = \int_{\text{spectrum}} d\mu_{\xi, \psi} f.
		\end{equation}
		\item Adjoints are given by complex conjugate functions: $f(T)^{\dagger} = \bar{f}(T).$
		\item We have $f(T) \bar{f}(T) = \bar{f}(T) f(T) = |f|^2(T).$
		\item The domain of $f(T) g(T)$ consists of vectors that are simultaneously in the domain of $g(T)$ and $(f\cdot g)(T)$; on this domain we have $f(T) g(T) = (f \cdot g)(T).$
		\item 
		If $a$ is a bounded operator that commutes with $T$ on every vector where both $aT$ and $Ta$ are defined, then $a$ commutes with $f(T)$ on every vector where both $a f(T)$ and $f(T) a$ are defined.
	\end{itemize}
	If, in addition to being self-adjoint, $T$ is positive (meaning it satisfies $\bra{\psi} T \ket{\psi} \geq 0$ for every $\ket{\psi}$), then its spectrum is contained in the range $[0, \infty).$
\end{theorem}

The last thing we will need to know about unbounded operators is the existence of polar decompositions.
To understand the statement of the theorem, it is helpful to know that a partial isometry from $\H$ to $\K$ is a map $u: \H \to \K$ for which $u^{\dagger} u$ and $u u^{\dagger}$ are both projections.
For a proof of the following theorem, see e.g. \cite[theorem 7.20]{weidmann2012linear}.
\begin{theorem} \label{thm:polar-decomposition}
	If $T$ is a densely defined, closed operator from $\H$ to $\K,$ then $T^{\dagger} T$ is positive and self-adjoint, and by the spectral theorem admits a positive square root $|T| = \sqrt{T^{\dagger} T}.$
	The operator $T$ can be written
	\begin{equation}
		T = u |T|
	\end{equation}
	for some partial isometry $u$ from $\H$ to $\K,$ and the projector $u^{\dagger} u$ projects onto the orthocomplement of $\ker(|T|),$ which is the same as the orthocomplement of $\ker(T).$
	Moreover, the polar decomposition is unique: given any decomposition
	\begin{equation}
		T = v P
	\end{equation}
	where $v$ is a partial isometry, $P$ is positive, and $v^{\dagger} v$ projects onto the orthocomplement of the kernel of $T$, we must have $v=u$ and $P=|T|.$
\end{theorem}

\begin{definition}
	Given a von Neumann algebra $\A,$ we say a closed operator $T$ is \textbf{affiliated with $\A$} if it commutes with each $a' \in \A'$ on every vector where both $a' T$ and $T a'$ are defined.
	This is the closest we can get to an unbounded operator being ``in'' a von Neumann algebra.
\end{definition}

\begin{theorem} \label{thm:neumann-polar}
	If $T$ is a bounded operator in the von Neumann algebra $\A,$ and its polar decomposition is $T = u |T|,$ then we have $u \in \A,$ and for every bounded function $f$ on the spectrum of $|T|$, we also have $f(|T|) \in \A.$
	
	If $T$ is a closed operator affiliated with the von Neumann algebra $\A$, and the polar decomposition of $T$ is $u |T|$, then we have $u \in \A,$ and for every bounded function $f$ on the spectrum of $|T|$, we also have $f(|T|) \in \A.$
\end{theorem}

This last theorem is easy to prove using the bicommutant theorem and the uniqueness of the polar decomposition.
It has the following useful corollary.

\begin{corollary} \label{cor:neumann-approximation}
	If $T$ is an unbounded operator affiliated with the von Neumann algebra $\A$, then there exists a sequence $T_n$ of operators in $\A$ such that, for each $\ket{\psi}$ in the domain of $T$, we have $T_n \ket{\psi} \to T \ket{\psi}.$
\end{corollary}
\begin{proof}
	Write $T = u |T|,$ and using the spectral theorem, take $T_n$ to be $u |T|_n$, where $|T|_n$ is the projection of $|T|$ onto the spectral range $[0, n].$
\end{proof}

\subsection{Analytic operator theory}
\label{sec:analytic-operators}

To study Tomita-Takesaki theory, it is essential to understand when an operator-valued function of the complex plane can be thought of as holomorphic.
The natural definition is the correct one --- given a function $f : \comps \to \B(\H)$, we say it is holomorphic at $z \in \comps$ if the limit
\begin{equation}
	f'(z) = \lim_{h \to 0} \frac{f(z+h) - f(z)}{h}
\end{equation} 
exists.
However, there is a subtlety: as explained in section \ref{subsec:operator-topologies}, there are many important topologies on spaces of operators, and this limit may exist with respect to some of those topologies and not with respect to others.
We therefore say, for example, that $f$ is \textit{weakly holomorphic} (or weakly analytic) at $z$ if this limit exists in the weak topology, and \textit{norm holomorphic} (or norm analytic) if this limit exists in the norm topology.
Note that by figure \ref{fig:topologies}, every norm analytic function is weakly analytic, but the converse is not guaranteed.\footnote{In fact one can show that for functions valued in spaces of bounded operators, norm analyticity and weak analyticity are equivalent, but we will not use this.
See \cite[section 9.24]{struatilua2019lectures}.}

Many of the basic theorems of complex analysis hold for analytic operator-valued functions, regardless of which topology is used.
The basic theorems of complex analysis all essentially follow from Cauchy's theorem that the integral of a holomorphic function around a simple closed curve vanishes.
So to understand complex analysis for operator-valued functions, it is necessary to develop a theory of operator-valued integration.
The theory of integration is easiest to understand for the norm topology on $\B(\H)$, where the theory is known as \textit{Bochner integration}.
This theory is explained in the textbooks \cite{dunford1988linear} and \cite{yosida2012functional}; see also my recent blog post \cite{sorce-blog-bochner}.
The point is that for any Banach space $X$ (e.g. $\B(\H)$ or $\H$ itself), and for any measure space $\Omega,$ there exists a class of integrable functions $f: \Omega \to X$ for which the integral $\int_{\Omega} d\mu f$ can be defined.
If $\Omega$ is a subset of some finite-dimensional real or complex space, then a continuous function is in the integrable class if and only if its norm has finite integral in the standard sense of integrals of real-valued functions.

The Bochner integral is linear, and it satisfies the triangle inequality
\begin{equation}
	\left\lVert \int_{\Omega} d\mu\, f \right\rVert \leq \int_{\Omega} d\mu\, \lVert f \rVert.
\end{equation}
The most important property of the Bochner integral is that if $Y$ is a Banach space and $\phi: X \to Y$ is a bounded linear functional of $X$, then we have
\begin{equation}
	\phi\left( \int d\mu\, f \right) = \int d\mu\, \phi(f).
\end{equation}
In particular, this is true for bounded linear functionals $\phi : X \to \comps.$
It is a basic fact of Banach space theory that a vector in Banach space is completely determined by the values it gives to bounded linear functionals.
(I.e., if $\phi(x)$ vanishes for every bounded linear functional, then $x$ is the zero vector.)
Consequently, the integral $\int d\mu\, f$ is completely determined by the integrals $\int d\mu\, \phi(f),$ which are ordinary integrals of complex-valued functions.
This lets one reduce the theory of contour integration in Banach space to the theory of contour integration for ordinary complex functions.
For more details, consult the references mentioned above; for our purposes, it is enough to know that if an operator- or vector-valued function of the complex plane is norm holomorphic, then its contour integral around any closed curve vanishes.
From this one can prove the residue theorem, Morera's theorem (that a function is holomorphic in a domain if all closed contour integrals vanish), and the Weierstrass theorem (that a uniformly convergent sequence of holomorphic functions converges to a holomorphic function).

Two essential properties of operator-valued complex analysis are as follows.
Detailed proofs can be found in \cite[sections 2.28, 2.30]{struatilua2019lectures} or \cite{sorce-blog-stone}.
The statements are sketched in figure \ref{fig:bounded-continuation}.
\begin{itemize}
	\item If a bounded, positive operator $T$ on $\H$ has spectrum bounded away from zero, then the function $z \mapsto T^z$ is norm analytic in the entire complex plane.
	This is easy to see, because it can be written in terms of an exponential as $T = e^{z \log{T}}.$
	\item If a bounded, positive operator $T$ on $\H$ has spectrum going all the way to zero, then the function\footnote{The operator $T^z$ is defined using the spectral theorem.
	While the function $x \mapsto x^{z}$ is only defined for positive $x$, the operator $T^z$ is defined in the case where $T$ has a nontrivial kernel by defining $x \mapsto x^z$ at zero by $0 \mapsto 0.$
	In other words, $T^z$ is defined to act as the zero operator on the kernel of $T$.} $z \mapsto T^z$ is norm analytic in the right half-plane and strongly continuous (but not norm continuous!) on the imaginary axis.
	This can be shown by projecting $T$ onto some subset of its spectrum that is bounded away from zero, using the previous bullet point, and taking a limit while applying the theorem that uniformly converging limits of holomorphic functions are holomorphic.
\end{itemize}

\begin{figure}
	\centering
	\includegraphics{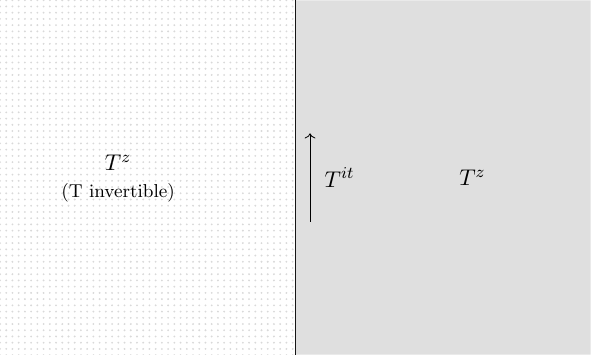}
	\caption{If $T$ is bounded and invertible, then the function $T^z$ is norm analytic in the entire complex plane.
	If it is bounded but not invertible, then $T^z$ is norm analytic in the right half-plane, and strongly continuous on the imaginary axis.}
	\label{fig:bounded-continuation}
\end{figure}

Another important tool for studying bounded operators using complex analysis is the resolvent integral.
If $T$ is a bounded self-adjoint operator, then its spectrum is some bounded subset of the real line.
Away from this subset, the function $z \mapsto (z - T)^{-1}$ is norm analytic.
One can show that if $f : \comps \to \comps$ is a function that is analytic in a neighborhood of the spectrum of $T$, then the operator $f(T)$ --- defined by the spectral theorem (theorem \ref{thm:spectral-theorem}) --- can be computed as a residue integral of the resolvent:
\begin{equation} \label{eq:bounded-residue}
	f(T)
		= \frac{1}{2 \pi i} \int_{\gamma} dz\, f(z) (z-T)^{-1}.
\end{equation}
In this equation, $\gamma$ is a simple, closed, counterclockwise contour surrounding the spectrum of $T$, and contained in the domain of analyticity of $f.$

The above considerations can be upgraded to unbounded operators.
If $T$ is a positive, self-adjoint, \textit{unbounded} operator, then one can always define the operators $T^z$ using the spectral theorem, but these operators will mostly be unbounded, and their domains will all be different.
It is therefore impossible to talk about the function $z \mapsto T^z$ being analytic in norm.
What one can say instead is that for certain subsets of the complex plane, there exist vectors $\ket{\psi}$ contained in the domain of $T^z$ for every $z$ in this subset, and that within these subsets of the complex plane, the functions $z \mapsto T^z \ket{\psi}$ are analytic.

Let $w = x + i y$ be a complex number, let $T$ be a positive, unbounded operator, and suppose that the vector $\ket{\psi}$ is in the domain of $T^{w}.$
Then, using the first bullet point of the spectral theorem (theorem \ref{thm:spectral-theorem}), it is not hard to show that $\ket{\psi}$ is in the domain of every $T^z$ where $z$ lies in the vertical strip between the imaginary axis and $w.$
See figure \ref{fig:vector-strip}.
The intuitive reason for this is that raising $T$ to an imaginary power produces a bounded operator (in fact, a partial isometry), so changing the imaginary part of $w$ does not affect the magnitude of the operator $T^w$, and moving the real part of $w$ closer to the imaginary axis makes $T^w$ less prone to diverge.
One can show that the function $z \mapsto T^z \ket{\psi}$ is holomorphic in the interior of this strip, and continuous at the boundary.\footnote{This fact, incidentally, is why holomorphic functions in a strip show up so often in functional analysis.}
For a proof, see for example \cite[section 9.15]{struatilua2019lectures}.

\begin{figure}
	\centering
	\includegraphics{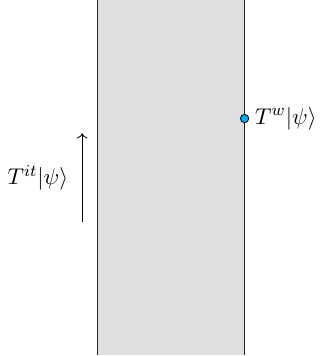}
	\caption{If $T$ is an unbounded positive operator, and the vector $\ket{\psi}$ is in the domain of $T^w,$ then it is also in the domain of $T^z$ for each $z$ in the vertical strip between $w$ and the imaginary axis.
	In the example sketched here, $w$ has positive real part and the imaginary axis is the left boundary of the strip.
	The function $z \mapsto T^z \ket{\psi}$ is holomorphic in the interior of this strip and continuous on the boundary.
	}
	\label{fig:vector-strip}
\end{figure}

Now, suppose that in addition to being positive, the unbounded operator $T$ has trivial kernel.
Even if the spectrum of $T$ contains zero, so that $T^{-1}$ is not invertible as a bounded operator, the spectral theorem (theorem \ref{thm:spectral-theorem}) implies that when the kernel of $T$ is trivial, the operator $f(T)$ is independent of the behavior of $f$ at zero.
So we may unambiguously define an unbounded operator $T^{-1}$ by applying the function $x \mapsto 1/x$ to $T$.
The domain of $T^{-1}$ is the closure of the image of $T$, and we have
\begin{equation}
	T T^{-1} = 1_{\D_{T^{-1}}}, \quad T^{-1} T = 1_{\D_T}.
\end{equation}
When $T$ is invertible, it is easy to show using the spectral theorem that $T^{it}$ is unitary, and so functions like $z \mapsto T^{z} \ket{\psi}$ are holomorphic in some strip and limit to unitary flows $it \mapsto T^{it} \ket{\psi}$ on the imaginary axis.
It is possible to show from this the following beautiful result, which allows us to completely determine the domain of $T$ in terms of analytic functions in a vertical strip.
\begin{theorem} \label{thm:vector-continuation}
	Let $T$ be a positive, invertible, unbounded operator.
	Then $\ket{\psi}$ is in the domain of $T^w$ if and only if, for each $\ket{\xi} \in \H,$ the map $it \mapsto \langle \xi | T^{it} \ket{\psi}$ admits an analytic continuation to the vertical strip between the imaginary axis and $w.$
	The overlap $\langle \xi | T^{w} | \psi \rangle$ is obtained by evaluating this analytic continuation at the point $w.$
\end{theorem}
\noindent See figure \ref{fig:weak-vector-strip} for intuition.
For a proof, see e.g. \cite[sections 9.15-9.20]{struatilua2019lectures}.

\begin{figure}
	\centering
	\includegraphics{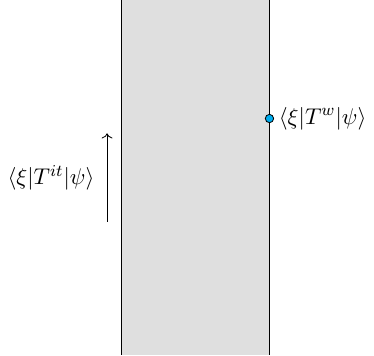}
	\caption{A sketch of the statement of theorem \ref{thm:vector-continuation}.
	To check that a vector $\ket{\psi}$ is in the domain of $T^{w},$ it is sufficient to check that all functions of the form $\langle \xi | T^{it} | \psi\rangle$ admit an analytic continuation from the imaginary axis to the vertical strip bounded by $w$.
	The overlap $\langle \xi | T^w |\psi \rangle$ is obtained by evaluating this analytic function at $w$.
	The main difference from figure \ref{fig:vector-strip} is that it is easier to study the analyticity of complex-valued functions like $\langle \xi | T^z | \psi \rangle$ than the analyticity of vector-valued functions like $T^z \ket{\psi}.$
	Note that once we know $\ket{\psi}$ is in the domain of $T^w$, it follows from figure \ref{fig:vector-strip} that the vector-valued function also admits an analytic continuation.}
	\label{fig:weak-vector-strip}
\end{figure}

The above discussion characterizes analytic extensions of unitary flows on vectors.
It will also be important to understand analytic extensions of unitary flows on operators.
\begin{theorem} \label{thm:operator-continuation}
	Let $T$ be a positive, invertible, unbounded operator.
	Let $a$ be a bounded operator.
	If $w$ is in the complex plane and the operator $T^{w} a T^{-w}$ is densely defined and bounded on its domain, then the operator $T^{z} a T^{-z}$ is densely defined and bounded on its domain for every $z$ in the vertical strip between the imaginary axis and w.
	Consequently, each $T^{z} a T^{-z}$ can be closed to a bounded operator.
	The function
	\begin{equation}
		F(z) = \bar{T^{z} a T^{-z}}
	\end{equation}
	is defined on the strip, and is norm analytic in the interior and strongly continuous on the boundary.
	(See figure \ref{fig:operator-strip}.)
	In fact, it is not necessary to check that $T^{w} a T^{-w}$ is bounded on its full domain; it suffices to check that $T^{w} a T^{-w}$ is defined and bounded on a core (definition \ref{def:core}) for the operator $T^{-w}.$
	
	Conversely, if $it \mapsto T^{it} a T^{-it}$ admits a norm-analytic continuation to the strip between the imaginary axis and $w,$ such that this analytic continuation is strongly continuous on the boundary of the strip, then $T^{w} a T^{-w}$ must be densely defined and bounded on its domain, and the analytic continuation is given by $z \mapsto \bar{T^{z}a T^{-z}}.$
\end{theorem}

\begin{figure}
	\centering
	\includegraphics{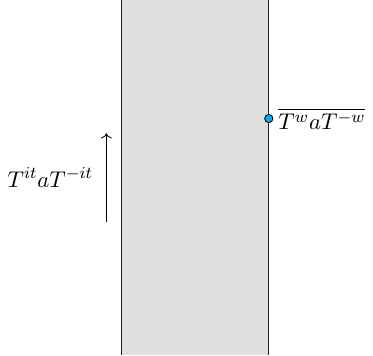}
	\caption{A sketch of the statement of theorem \ref{thm:operator-continuation}.
	Given an unbounded, invertible, positive operator $T$, a bounded operator $a,$ and a complex number $w$ such that $T^w a T^{-w}$ is densely defined and bounded on its domain, the function on the strip between the imaginary axis and $w$ given by $z \mapsto \bar{T^z a T^{-z}}$ is norm analytic in the interior of the strip and strongly continuous on its boundary.
	On the imaginary axis, it takes the values $T^{it} a T^{-it}.$}
	\label{fig:operator-strip}
\end{figure}

\section{Uniqueness of thermal symmetries}
\label{sec:uniqueness}

In a quantum system admitting density matrices, as explained in the introduction, the density matrix $\rho$ is said to be thermal with respect to the Hamiltonian $H$ at inverse temperature $\beta$ if it has the form 
\begin{equation}
	\rho_{\beta} = \frac{e^{-\beta H}}{\tr(e^{-\beta H})} \equiv \frac{e^{-\beta H}}{Z}.
\end{equation}
In order for this density matrix to make sense, the spectrum of $H$ must be bounded below, and must be sufficiently tame so that $e^{- \beta H}$ has finite trace.

The time-evolved two-point function of the thermal state can be written as
\begin{equation} \label{eq:lattice-thermal-2pt}
	\langle e^{i H t} a e^{-i H t} b \rangle_{\beta}
		= \frac{1}{Z} \tr(e^{- (\beta - i t) H} a e^{-i H t} b).
\end{equation}
Kubo, Martin, and Schwinger (KMS) \cite{kubo1957statistical, martin1959theory} observed that this function admits an analytic continuation from $it$ to more general complex $z.$
Intuitively, we would like to simply make the substitution $it \mapsto z$, and write our analytic continuation as
\begin{equation}
	F(z) = \frac{1}{Z} \tr(e^{-(\beta-z) H} a e^{- z H} b).
\end{equation}
If $H$ is bounded, then this is an analytic function of the entire complex plane.
If $H$ is only bounded below, however, then this function is not even well defined for arbitrary $z$.
For $\Re(z) < 0$ the operator $e^{-z H}$ diverges, and for $\Re(z) > \beta$ the operator $e^{-(\beta - z) H}$ diverges.
However, in the vertical strip of the complex plane given by $0 \leq \Re(z) \leq \beta,$ the function $F(z)$ is well defined, and since it is written in terms of exponentials, it is analytic.

The function $F(z)$ therefore furnishes an analytic continuation of the two-point function to the strip of width $\beta.$
Furthermore, the boundary values of this analytic continuation are given by
\begin{align}
	F(it)
		& = \langle e^{i H t} a e^{- i H t} b \rangle_{\beta}
\end{align}
and
\begin{align}
	\begin{split}
	F(\beta + i t)
		& = \frac{1}{Z} \tr(e^{i H t} a e^{- \beta H} e^{- i H t} b) \\
		& = \frac{1}{Z} \tr(e^{- \beta H} b e^{i H t} a e^{- i H t}) \\
		& = \langle b e^{i H t} a e^{-i H t} \rangle_{\beta}.
	\end{split}
\end{align}
See figure \ref{fig:KMS}.

\begin{figure}
	\centering
	\includegraphics{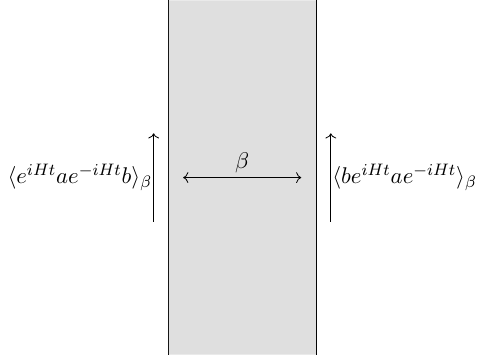}
	\caption{In a thermal state of inverse temperature $\beta,$ the time-evolved two-point function can be analytically continued to a function in a vertical strip of width $\beta.$
	On the right edge of the strip, the function is equal to the ``flipped'' time-evolved two-point function.}
	\label{fig:KMS}
\end{figure}

We would like to use the above observations to produce a definition of thermality that makes sense in settings where there are no density matrices.
To do this, we will characterize thermality of a state $\ket{\Omega}$ within the algebra $\A$ in terms of the analytic structure of the two-point functions $\bra{\Omega} a b \ket{\Omega}$ with $a, b \in \A$.
So far, we have written everything in terms of two-point functions of the form $\tr(\rho_{\beta} a b)$; to generalize, we must rewrite these correlators in terms of a pure state.
In other words, we must purify the Gibbs ensemble.

A convenient and familiar purification for the Gibbs ensemble is provided by the ``thermofield double'' state, which is a special case of the canonical purification \cite{Dutta:2019gen}, itself a special case of the GNS purification (see e.g. \cite[chapter 7]{conway2000course}) for algebraic states.
The thermofield double purification is constructed by doubling the Hilbert space $\H$ to $\H \otimes\H'$ and defining the state $\ket{\text{TFD}_{\beta}} \in \H\otimes \H'$ by
\begin{equation}
	\ket{\text{TFD}_{\beta}}
		= \frac{1}{\sqrt{Z}} \sum_{n} e^{- \beta E_n/2} \ket{E_n}_\H \otimes \ket{E_n'}_{\H'}.
\end{equation}
For operators $a$ and $b$ acting on $\H,$ the two-point function in this pure state is given by
\begin{equation}
	\bra{\text{TFD}_{\beta}} a b \ket{\text{TFD}_{\beta}} = \tr(\rho_{\beta} a b) = \langle a b \rangle_{\beta}.
\end{equation}
It is straightforward to check that $\ket{\text{TFD}_{\beta}}$ is fixed by time evolution with respect to the operator $H - H',$ and that conjugating an operator on $\H$ by $e^{i (H - H') t}$ is the same as conjugating it by $e^{i H t}.$
So the KMS observation can be written in terms of two-sided, pure-state quantities by saying that for operators $a$ and $b$ acting on $\H$, there is a Hamiltonian $H - H'$ acting on $\H \otimes \H',$ and a state $\ket{\text{TFD}_{\beta}}$ in $\H \otimes \H',$ with the following properties:
\begin{itemize}
	\item $\ket{\text{TFD}_{\beta}}$ is a purification of the thermal state on $\H$ at inverse temperature $\beta.$
	\item $e^{i (H - H') t} \ket{\text{TFD}_{\beta}} = \ket{\text{TFD}_{\beta}}.$
	\item For $a \in \B(\H) \otimes 1_{\H'}$, we have $e^{i (H - H') t} a e^{- i (H - H') t} \in \B(\H) \otimes 1_{\H'}$.
	\item For $a, b \in \B(\H) \otimes 1_{\H'},$ the time-evolved two-point function
	\begin{equation}
		\bra{\text{TFD}_{\beta}} e^{i (H - H') t} a e^{- i (H - H') t} b \ket{\text{TFD}_{\beta}}
	\end{equation}
	admits an analytic continuation of KMS form, in that there exists a function $F(z)$ in the strip of width $\beta$ with boundary values
	\begin{equation}
		F(it)
			= \bra{\text{TFD}_{\beta}} e^{i (H - H') t} a e^{- i (H - H') t} b \ket{\text{TFD}_{\beta}}
	\end{equation}
	and
	\begin{equation}
		F(\beta+i t)
			= \bra{\text{TFD}_{\beta}} b e^{i (H - H') t} a e^{- i (H - H') t} \ket{\text{TFD}_{\beta}}
	\end{equation}
\end{itemize}

The above observations motivate the following abstract definition of thermality for quantum systems described by von Neumann algebras.
Note that in the definition given below, the operator that was called $H - H'$ in the lattice setting has been renamed $H$; there is no longer any need to reserve labels for Hamiltonians on individual subsystems, since in the general setting these ``one-sided'' Hamiltonians do not always exist.
\begin{definition}[KMS condition] \label{def:KMS}
	Let $\H$ be a Hilbert space, $\ket{\Omega}$ a state, and $\A$ a von Neumann algebra.
	Let $H$ be a self-adjoint, possibly unbounded operator.
	The state $\ket{\Omega}$ is said to satisfy the \textit{KMS condition} with respect to $\A$ and $H$  at inverse temperature $\beta$ if the following three properties hold.
	\begin{enumerate}[(i)]
		\item
		$H$ generates a symmetry of the state: for every real number $t,$ we have
		\begin{equation}
			e^{- i H t} \ket{\Omega} = \ket{\Omega}.
		\end{equation}
		\item 
		$H$ generates an automorphism of the algebra: for every real number $t$ and every $a \in \A,$ we have
		\begin{equation}
			e^{i H t} a e^{-i H t} \in \A.
		\end{equation}
		\item
		Two-point functions of $\A$ in the state $\ket{\Omega}$ look thermal with respect to the flow generated by $H$: for every $a, b \in \A,$ the function
		\begin{equation}
			F(it) = \bra{\Omega} e^{i H t} a e^{-i H t} b \ket{\Omega}
				= \bra{\Omega} a e^{- i H t} b \ket{\Omega}
		\end{equation}
		admits a bounded analytic continuation to the vertical strip $0 \leq \Re(z) \leq \beta,$ and on the right boundary of this strip the analytic continuation is given by
		\begin{equation}
			F(\beta + i t)
				= \bra{\Omega} b e^{i H t} a e^{- i H t} \ket{\Omega}
				= \bra{\Omega} b e^{i H t} a \ket{\Omega}.
		\end{equation}
	\end{enumerate}
\end{definition}

The point of this definition is that it has stripped away everything having to do with density matrices and tensor product decompositions.
It lets us talk about thermal physics in systems without density matrices, such as quantum field theories, by expressing thermality as a property of the analytic structure of two-point functions.
I find it helpful to think in terms of the following slogan.
\begin{quote}
	Thermality on the lattice is always defined with respect to some Hamiltonian, and therefore with respect to some time-evolution map.
	The KMS condition puts the time-evolution map front and center: it requires that two-point functions evolved with respect to a given time-evolution map have the same structure as the two-point functions of a Gibbs state in a lattice system evolved with respect to the lattice Hamiltonian.
\end{quote}
Given a generic state $\ket{\Omega}$ and a generic von Neumann algebra $\A,$ two natural questions arise:
\begin{itemize}
	\item Does there exist a Hamiltonian for which the state $\ket{\Omega}$ looks thermal in the system $\A$?
	\item Could there exist multiple different Hamiltonians with this property?
\end{itemize}
The first question is much harder to answer than the second.
Using Tomita-Takesaki theory, we will ultimately see that the answer is \textit{yes}, at least when the state $\ket{\Omega}$ is cyclic and separating (defined in section \ref{sec:cyclic-separating}).
But to motivate the introduction of Tomita-Takesaki theory, we will begin by answering the second question.
We will show that if $\ket{\Omega}$ and $\A$ satisfy the KMS condition with respect to \textit{some} Hamiltonian $H$, then that Hamiltonian \textit{must} be the Tomita-Takesaki modular Hamiltonian $K = - \log \Delta.$
This motivates the introduction of the modular operator directly from thermal physics, rather than introducing it as an abstract mathematical tool.

\begin{theorem}
	Let $\A$ be a von Neumann algebra, and let $\ket{\Omega}$ be a cyclic and separating state for $\A.$
	Suppose that $H$ is a self-adjoint operator such that $\A$, $\ket{\Omega}$, and $H$ satisfy the KMS condition with $\beta=1.$
	Then $e^{-H/2}$ must be the positive part of the closed, antilinear operator $S$ that acts on $\A \ket{\Omega}$ as
	\begin{equation}
		S (a \ket{\Omega})
			= a^{\dagger} \ket{\Omega}.
	\end{equation}
	Consequently, $H$ is given by $- \log \Delta,$ where $\Delta = S^{\dagger} S$ is the Tomita-Takesaki modular operator.
\end{theorem}
\begin{proof}
	Fix two operators $a, b \in \A.$
	By the KMS condition, there exists a function $F(z),$ defined in the vertical strip $0 \leq \Re(z) \leq 1,$ analytic in the interior of the strip and continuous on its boundary, and with boundary values
	\begin{align}
		F(it)
			& = \bra{\Omega} a e^{- i H t} b \ket{\Omega}, \label{eq:left-boundary-KMS} \\
		F(1 + i t)
			& = \bra{\Omega} b e^{i H t} a \ket{\Omega}.
	\end{align}
	See figure \ref{fig:concrete-KMS}.
	
	\begin{figure}
		\centering
		\includegraphics{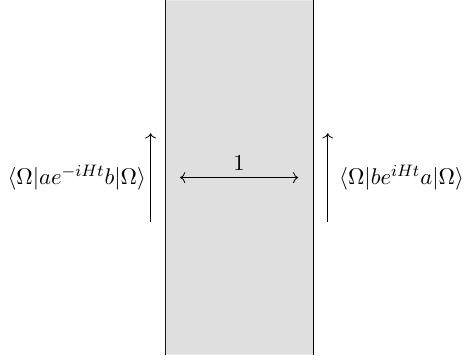}
		\caption{Given a vector $\ket{\Omega}$ and an algebra $\A$ satisfying the KMS condition for the Hamiltonian $H$, the function  $\bra{\Omega} a e^{-i H t} b \ket{\Omega}$ can be analytically continued to a vertical strip of width $1,$ with the analytic continuation being equal to $\bra{\Omega} b e^{i H t} a \ket{\Omega}$ on the right edge of the strip.}
		\label{fig:concrete-KMS}
	\end{figure}
	
	Intuitively, we would expect the function $F(z)$ to be obtained by substituting $it \mapsto z$ in equation \eqref{eq:left-boundary-KMS}; that is, we would like to have a formula like
	\begin{equation}
		F(z)
			= \bra{\Omega} a e^{- z H} b \ket{\Omega}.
	\end{equation}
	Unfortunately, this expression does not necessarily make sense; the Hamiltonian $H$ is not assumed to be bounded below, so $e^{- z H}$ is an unbounded operator, and there is no guarantee that $b \ket{\Omega}$ is in the domain of $e^{-z H}.$
	If we take this expression as a guiding principle, however, and evaluate it at $z=1,$ then by the KMS condition we expect there to be some sense in which the operator $H$ satisfies
	\begin{equation} \label{eq:KMS-swap}
		\langle a^{\dagger} \Omega| e^{- H} | b \Omega \rangle
			= \langle  b^{\dagger} \Omega | a \Omega \rangle. 
	\end{equation}
	Again, this expression is not strictly correct, but it gives us an intuitive understanding of why the operator $H$ is uniquely determined by the KMS condition --- cyclicity of the vector $\ket{\Omega}$ implies that vectors of the form $a \ket{\Omega}$ are dense in $\H$, so equation \eqref{eq:KMS-swap} provides a complete set of matrix elements for the operator $e^{- H}.$
	
	Now, let us take our ``morally correct'' expression \eqref{eq:KMS-swap} and rewrite it as
	\begin{equation} \label{eq:antilinear-H}
		\langle e^{-H/2} a^{\dagger} \Omega | e^{-H/2} b \Omega \rangle
			= \langle b^{\dagger} \Omega | a \Omega \rangle.
	\end{equation}
	This expression is interesting because it is reminiscent of the identities satisfied by antilinear operators.
	An antilinear operator can be thought of as a map from $\H$ to the complex conjugate space $\bar{\H},$ which is a Hilbert space with the same vectors as $\H$ but with inner product
	\begin{equation}
		\braket{\psi}{\xi}_{\bar{\H}}
			= \braket{\xi}{\psi}_{\H}.
	\end{equation}
	If $e^{-H/2}$ were an antilinear operator mapping $a^{\dagger} \ket{\Omega}$ to $a \ket{\Omega}$ and $b \ket{\Omega}$ to $b^{\dagger} \ket{\Omega},$ then we would have
	\begin{equation}
		\langle e^{-H/2} a^{\dagger} \Omega | e^{-H/2} b \Omega \rangle_{\bar{\H}}
			= \langle a \Omega | b^{\dagger} \Omega \rangle_{\bar{\H}}
			= \langle b^{\dagger} \Omega | a \Omega \rangle_{\H},
	\end{equation}
	which is exactly equation \eqref{eq:antilinear-H}.
	
	Now, while $e^{-H/2}$ is not an antilinear operator, the above considerations suggest that it is worthwhile to define and study an antilinear operator that maps $a^{\dagger} \ket{\Omega}$ to $a \ket{\Omega}$ and $b \ket{\Omega}$ to $b^{\dagger} \ket{\Omega}.$
	We define the antilinear operator $S_0$ on the domain $\D_{S_0} = \A \ket{\Omega}$ by
	\begin{equation}
		S_0 x \ket{\Omega} = x^{\dagger} \ket{\Omega}, \quad x \in \A.
	\end{equation}
	Because $\ket{\Omega}$ is cyclic, the operator $S_0$ is densely defined, but generally unbounded.
	One can show (see appendix \ref{app:tomita}) that because $\ket{\Omega}$ is cyclic and separating, the operator $S_0$ has the following properties.
	\begin{itemize}
		\item $S_0$ is preclosed.
		Its closure is denoted $S,$ and is called the \textbf{Tomita operator}.
		\item Denoting the polar decomposition of $S$ by
		\begin{equation}
			S
				= J \Delta^{1/2},
		\end{equation}
		the antilinear partial isometry $J$ is called the \textbf{modular conjugation} and the operator $\Delta$ is called the \textbf{modular operator}.
		The modular operator has trivial kernel, and is therefore invertible.
		The operator $K = - \log \Delta$ is called the \textbf{modular Hamiltonian}.
		\item The antilinear partial isometry $J$ is in fact an antiunitary operator, and satisfies $J^2 = 1,$ hence $J = J^{\dagger}.$
		\item
		The domain of $S$, which is the same as the domain of $\Delta^{1/2}$, consists of all vectors of the form $T \ket{\Omega}$ such that either (i) $T$ is in $\A,$ or (ii) $T$ is affiliated to $\A$ and $\ket{\Omega}$ is in the domain of both $T$ and $T^{\dagger}.$
		In either case, we have $S T \ket{\Omega} = T^{\dagger} \ket{\Omega}.$
	\end{itemize}
	Now, for every $a, b \in \A,$ the vectors $a^{\dagger} \ket{\Omega}$ and $b \ket{\Omega}$ are in the domain of $\Delta^{1/2} = e^{-K/2}.$
	So the operator $K$ satisfies
	\begin{align} \label{eq:exact-K-swap}
		\begin{split}
		\langle e^{-K/2} a^{\dagger} \Omega | e^{-K/2} b \Omega \rangle
			& = \langle J^2 e^{-K/2} a^{\dagger} \Omega | e^{-K/2} b \Omega \rangle \\
			& = \langle J e^{-K/2} b \Omega | J e^{-K/2} a^{\dagger} \Omega \rangle \\
			& = \langle S b \Omega | S a^{\dagger} \Omega \rangle \\
			& = \langle b^{\dagger} \Omega | a \Omega \rangle.
		\end{split}
	\end{align}
	This expression is exactly correct, and matches the ``morally correct'' expression \eqref{eq:antilinear-H} for the operator $H.$
	We will now use this observation to show $\Delta=e^{-H},$ hence $K=H.$
	
	Note first that both $\Delta$ and $e^{-H}$ are self-adjoint operators.
	It therefore suffices to show that $e^{-H}$ is an extension of $\Delta$, since the inclusion
	\begin{equation}
		\Delta \subseteq e^{-H}
	\end{equation}
	implies\footnote{Here we have used the implication $T \subseteq S \Rightarrow S^{\dagger} \subseteq T^{\dagger}$, which was stated in proposition \ref{prop:adjoint-properties}.}
	\begin{equation}
		e^{-H} = (e^{-H})^{\dagger} \subseteq \Delta^{\dagger} = \Delta,
	\end{equation}
	and hence $\Delta = e^{-H}.$
	So at a concrete level, the remaining proof reduces to showing that for every vector $\ket{\psi}$ in the domain of $\Delta,$ the vector $\ket{\psi}$ is also in the domain of $e^{-H},$ and we have $\Delta \ket{\psi} = e^{-H} \ket{\psi}.$
	
	Let us consider an arbitrary vector in the domain of $\Delta$.
	As explained in section \ref{sec:analytic-operators}, any vector in the domain of $\Delta$ is also in the domain of $\Delta^{1/2}.$
	As explained in the bulleted list above, any vector in the domain of $\Delta^{1/2}$ may be written as $T \ket{\Omega},$ where $T$ is affiliated to $\A$ and where $\ket{\Omega}$ is in the domain of both $T$ and $T^{\dagger}.$
	We aim to show that $T \ket{\Omega}$ is in the domain of $e^{-H}.$
	By theorem \ref{thm:vector-continuation}, this is equivalent to showing that for every $\ket{\xi} \in \H,$ the function
	\begin{equation}
		it \mapsto \langle \xi| e^{- i H t} |T \Omega\rangle
	\end{equation}
	admits an analytic continuation to the strip $0 \leq \Re(z) \leq 1.$
	If $\ket{\xi}$ is of the form $a^{\dagger} \ket{\Omega}$ for $a \in \A,$ then this analytic continuation exists by the KMS condition.\footnote{Technically we must show that the KMS condition holds for affiliated operators, but this can be shown to follow from the KMS condition for bounded operators via corollary \ref{cor:neumann-approximation}.}
	The general case follows by taking limits using cyclicity of $\ket{\Omega}.$
	So $T \ket{\Omega}$ is in the domain of $e^{-H},$ and the action of $e^{-H}$ on that vector is given (via the KMS condition and theorem \ref{thm:vector-continuation}) by the formula
	\begin{equation}
		\langle a^{\dagger} \Omega| e^{- H} |T \Omega\rangle
			= \langle T^{\dagger} \Omega | a \Omega \rangle.
	\end{equation}
	But we also have
	\begin{equation}
		\langle a^{\dagger} \Omega | \Delta | T \Omega \rangle
			= \langle \Delta^{1/2} a^{\dagger} \Omega | \Delta^{1/2} T \Omega \rangle
			= \langle T^{\dagger} \Omega | a \Omega \rangle.
	\end{equation}
	Since vectors of the form $a^{\dagger}\ket{\Omega}$ are dense in Hilbert space,
	These two identities give the vector equation
	\begin{equation}
		e^{-H} |T \Omega \rangle
			= \Delta |T \Omega \rangle.
	\end{equation}
	So every vector in the domain of $\Delta$ is in the domain of $e^{-H},$ and the actions of $\Delta$ and $e^{-H}$ agree on such vectors.
	As explained above, this completes the proof of $\Delta = e^{-H}.$
\end{proof}

So far, we have shown that if there exists a Hamiltonian $H$ satisfying the KMS condition (definition \ref{def:KMS}), then it must be the modular Hamiltonian $K$.
It is important to note that this does not tell us whether or not the modular Hamiltonian \textit{does} satisfy the KMS condition; all we know so far is that the modular operator is the only operator that could work in principle.
It is easy to see from the identity $\Delta = S^{\dagger} S$ that we have $\Delta \ket{\Omega} = \ket{\Omega},$ so the modular Hamiltonian certainly satisfies property (i) of definition \ref{def:KMS}.
The second property is the hard one to show --- we must show that for $a \in \A,$ we have
\begin{equation}
	e^{i K t} a e^{- i K t} = \Delta^{-it} a \Delta^{it} \in \A.
\end{equation}
This is Tomita's theorem; the proof is quite involved, and is the subject of the next section.
To convince ourselves that it is actually worthwhile to prove this theorem, however, let us show that once we have proved Tomita's theorem, it will immediately follow that $K$ satisfies condition (iii) of the KMS condition, and therefore does indeed provide a completely general, unique, thermal time-evolution map for an out-of-equilibrium state.

\begin{theorem}[Modular KMS]
	Let $\H$ be a Hilbert space, $\A$ a von Neumann algebra, $\ket{\Omega}$ a cyclic and separating state, $\Delta$ the modular operator, and $K = - \log \Delta$ the modular Hamiltonian.
	Suppose Tomita's theorem is true, so that we have
	\begin{equation}
		e^{i K t} a e^{-i K t} \in \A, \quad a \in \A.
	\end{equation}
	Then $K$ satisfies the third condition of the KMS condition.
	I.e., for every $a, b \in \A$, the function
	\begin{equation}
		F(it) = \bra{\Omega} a e^{-i K t} b \ket{\Omega}
	\end{equation}
	admits a bounded analytic continuation to the vertical strip $0 \leq \Re(z) \leq 1,$ and on the right boundary of this strip the analytic continuation is given by
	\begin{equation}
		F(1 + i t) = \bra{\Omega} b e^{i K t} a \ket{\Omega}.
	\end{equation}
\end{theorem}
\begin{proof}
	According to the properties of the modular operator described in appendix \ref{app:tomita}, the states $b \ket{\Omega}$ and $a^{\dagger} \ket{\Omega}$ are in the domain of the operator $\Delta^{1/2} = e^{-K/2}.$
	So, as explained in section \ref{sec:analytic-operators} around figure \ref{fig:vector-strip}, these states are in the domain of $\Delta^w$ for every $0 \leq \Re(w) \leq 1/2,$ and the functions
	\begin{equation}
		w \mapsto \Delta^{w} b\ket{\Omega}
	\end{equation} 
	and
	\begin{equation}
		w \mapsto \Delta^{w} a^{\dagger} \ket{\Omega}
	\end{equation}
	are bounded and analytic in this strip.
	Consequently, the function
	\begin{equation}
		w \mapsto \langle \Delta^{\bar{w}} a^{\dagger} \Omega | \Delta^{w} b \Omega \rangle
	\end{equation}
	is bounded and analytic in the strip $0 \leq \Re(w) \leq 1/2,$ and the function
	\begin{equation}
		F(z)
			= \langle \Delta^{\bar{z}/2} a^{\dagger} \Omega | \Delta^{z/2} b \Omega \rangle
	\end{equation}
	is bounded and analytic in the strip $0 \leq z \leq 1.$
	
	On the left side of the strip, we have
	\begin{equation}
		F(it) = \langle \Delta^{-it/2} a^{\dagger} \Omega| \Delta^{it/2} b \Omega \rangle
		= \langle \Omega | a \Delta^{it} b | \Omega \rangle
		= \langle \Omega | a e^{- i K t} b | \Omega \rangle.
	\end{equation}
	So this function has the correct KMS boundary value on the left side of the strip.
	On the right side of the strip, we have
	\begin{align}
		\begin{split}
		F(1+it)
			& = \langle \Delta^{1/2} \Delta^{-it/2} a^{\dagger} \Omega | \Delta^{1/2} \Delta^{it/2} b \Omega \rangle \\
			& = \langle \Delta^{1/2} (\Delta^{-it/2} a^{\dagger} \Delta^{it/2}) \Omega | \Delta^{1/2} (\Delta^{it/2} b \Delta^{-it/2}) \Omega \rangle.
		\end{split}
	\end{align}
	Since we are assuming that Tomita's theorem holds, we have
	\begin{equation}
		\Delta^{-it/2} a^{\dagger} \Delta^{it/2} \in \A
	\end{equation}
	and
	\begin{equation}
		\Delta^{it/2} b \Delta^{-it/2} \in \A.
	\end{equation}
	So we may use equation \eqref{eq:exact-K-swap} to obtain
	\begin{align}
		\begin{split}
			F(1+it)
			& = \langle (\Delta^{it/2} b^{\dagger} \Delta^{-it/2}) \Omega | (\Delta^{-it/2} a \Delta^{it/2}) \Omega \rangle \\
			& = \langle b^{\dagger} \Omega | \Delta^{-it} | a \Omega \rangle \\
			& = \langle \Omega | b e^{i K t} a | \Omega \rangle.
		\end{split}
	\end{align}
	So the function $F(z)$ provides an analytic continuation of the two-point function satisfying the KMS condition.
\end{proof}

\section{Tomita's theorem, tidy operators, and the existence of thermal symmetries} \label{sec:tomitas-theorem}

Hopefully by now it is clear why we should care so much about proving Tomita's theorem.
In the previous section, we saw not only that the modular Hamiltonian is the unique Hamiltonian that \textit{could} provide a general notion of thermal time, but also that Tomita's theorem is the only obstacle to guaranteeing that it \textit{does}.

This section presents a proof of Tomita's theorem.
The idea of the proof is to study analytic continuations of maps like
\begin{equation} \label{eq:operator-flow}
	it \mapsto \Delta^{-it} a \Delta^{it}
\end{equation}
for $a \in \A.$
Since analytic functions are highly constrained, we will be able to use analyticity to show explicitly that all commutators of the form $[\Delta^{-it} a \Delta^{it}, b']$ vanish when $b'$ is an operator in the commutant.
By the bicommutant theorem (cf. section \ref{subsec:operator-topologies}), this will imply $a \in \A.$

The structure of analytic continuations of maps like \eqref{eq:operator-flow} is described by theorem \ref{thm:operator-continuation}.
In particular, that theorem implies that the map in equation \eqref{eq:operator-flow} admits an analytic extension to the entire complex plane if and only if, for every integer $n$, the operator $\Delta^{-n} a \Delta^n$ is densely defined and bounded on its domain.
This is a big demand, and in fact we should not expect it to be true for arbitrary $a \in \A$ except in very special situations where the modular operator $\Delta$ and its inverse $\Delta^{-1}$ are bounded.
So the way we will actually proceed is by constructing some dense subspace of $\A$ for which each $\Delta^{-n} a \Delta^n$ is bounded, and in fact for which the norm of $\Delta^{-n} a \Delta^n$ is bounded at infinity by an exponential function of $n.$
I will call this the \textbf{tidy subspace} of $\A$, denoted $\A_{\text{tidy}}$, and its members will be called \textbf{tidy operators}.
I will prove Tomita's theorem for tidy operators by studying the operators $\Delta^{-n} a \Delta^{n}$ and applying Carlson's theorem, then obtain the general result by continuity.

The idea behind constructing tidy operators is that they should be operators for which the modular operator and its inverse ``look bounded.''
Given $0 < \lambda_1 < \lambda_2,$ the Heaviside theta function
\begin{equation}
	\Theta_{[\lambda_1, \lambda_2]}(\Delta) \equiv \Theta((\lambda_2 - \Delta)(\Delta - \lambda_1))
\end{equation}
truncates the spectrum of $\Delta$ to the range $[\lambda_1, \lambda_2].$
We will obtain tidy operators by starting with a vector like $x \ket{\Omega}$ for some arbitrary $x \in \A,$ then acting on this vector with $\Theta_{[\lambda_1, \lambda_2]}(\Delta).$
The resulting vector has no support in the spectral subspaces of $\Delta$ near zero and infinity.
We will show that this vector can be written as some other operator in $\A$ acting on $\ket{\Omega}$ via an equation like
\begin{equation}
	\Theta_{[\lambda_1, \lambda_2]}(\Delta) x \ket{\Omega} = x_{[\lambda_1, \lambda_2]} \ket{\Omega},
\end{equation} 
and take $x_{[\lambda_1, \lambda_2]}$ to be one of our tidy operators.

Unfortunately, it takes some effort to get control over the operators $\Theta_{[\lambda_1 ,\lambda_2]}(\Delta).$
It is much easier to control analytic functions of $\Delta,$ since these can sometimes be expressed as contour integrals using the resolvent of $\Delta$.
We will therefore proceed by developing a theory of analytic mollifiers, which are operators $f(\Delta)$ for $f$ an analytic function in a neighborhood of $[0, \infty)$ such that $f$ vanishes at infinity.
We will obtain tidy operators by approximating step functions using analytic mollifiers.

It may be interesting to note that while Zsid\'{o}'s proof \cite{zsido-proof} constructed analytic mollifiers like the ones I construct here, it was necessary in that case to appeal to the theory of analytic generators \cite{cioranescu1976analytic} to finish proving Tomita's theorem. 
This is because while one can use analytic mollifiers to produce operators $a \in \A$ for which $\Delta^{-it} a \Delta^{it}$ admits an entire analytic continuation, these analytic continuations are superexponentially growing at infinity,\footnote{See e.g. \cite[section 10.22]{struatilua2019lectures} for a discussion of when the analytic continuation of modular flow is exponentially bounded.} and cannot be constrained using Carlson's theorem.

In section \ref{sec:resolvent-lemma}, I present a proof of a lemma due to Takesaki \cite{takesaki2006tomita} concerning what happens when the resolvent of the modular operator is used as a mollifier.
In section \ref{sec:tidy-subspace}, I present a construction of the tidy subspace $\A_{\text{tidy}}$ by extending arguments made in \cite{zsido-proof}.
In section \ref{sec:main-proof}, I show that for $a$ a tidy operator and $z \in \comps,$ the operator $\Delta^{-z} a \Delta^{z}$ is densely defined and bounded on its domain, and furthermore that when $z$ is an integer, the closure of $\Delta^{-z} a \Delta^z$ lies in $\A.$
I then use Carlson's theorem to show $\Delta^{-it} a \Delta^{it} \in \A$ for $a$ in the tidy subspace, and obtain the statement for general $a \in \A$ by continuity.
In section \ref{sec:modular-conjugation}, I show that the techniques developed in this section also lead to an easy proof of the statement that for $a \in \A$ and $J$ the modular conjugation, the operator $J a J$ is in $\A'.$

\subsection{Takesaki's resolvent lemma}
\label{sec:resolvent-lemma}

The following lemma is extremely important in the development of the ensuing theory.
Unfortunately, its proof is not particularly instructive --- it involves a mathematical trick without any obvious physical meaning.
Nevertheless, I hope that the reasons for trying to prove such a lemma are apparent from the preceding discussion, so the lack of insight provided by the proof will be acceptable, if undesirable.

\begin{lemma}[Takesaki's resolvent lemma] \label{lem:resolvent-lemma}
	Let $\A$ be a von Neumann algebra, $\ket{\Omega}$ a cyclic and separating vector, and $\Delta$ the associated modular operator.
	Fix $x' \in \A'.$
	Then for every complex number $z$ such that $(z - \Delta)$ is invertible as a bounded operator, we have
	\begin{equation}
		(z - \Delta)^{-1} x' \ket{\Omega} = x \ket{\Omega}
	\end{equation}
	for some unique $x \in \A$ which depends on $z$.
	Furthermore, the norm of this operator satisfies the bound
	\begin{equation}
		\lVert x \rVert_{\infty}
			\leq \frac{\lVert x' \rVert_{\infty}}{\sqrt{2 (|z| - \Re(z))}}.
	\end{equation}
\end{lemma}
\begin{proof}
	First note that because $(z - \Delta)^{-1} x' \ket{\Omega}$ is in the domain of $z - \Delta,$ it is in the domain of $\Delta,$ and so in the domain of $\Delta^{1/2}.$
	It follows (cf. appendix \ref{app:tomita}) that it can always be written in the form $x \ket{\Omega}$ for some operator $x$ \textit{affiliated} to $\A.$
	The whole work of the lemma is in showing that $x$ is bounded.
	
	To do this, it would suffice to show that $x$ has bounded action on vectors of the form $b' \ket{\Omega}$ for $b' \in \A'.$
	Unfortunately, showing this does not seem to be tractable in general.
	Instead, we will write the polar decomposition of $x$ as
	\begin{equation}
		x = u |x|,
	\end{equation}
	and try to show that $x$ is bounded when acting on vectors of the form $P \ket{\Omega}$ where $P$ is a spectral projection of $|x|.$
	In fact, because of the way the modular operator shows up in the lemma statement, it will be easier to show that $x^{\dagger}$ is bounded when acting on vectors of the form $Q \ket{\Omega}$ for $Q$ a spectral projection of $|x^{\dagger}|.$
	Once we have shown this, we will apply the fact that $\ket{\Omega}$ is separating to conclude that $x$ is bounded, and derive the specific bound given in the statement of the lemma.
	
	Let $I$ be an interval in $[0, \infty),$ and let $\Pi_I(|x^{\dagger}|)$ denote the spectral projection of $|x^{\dagger}|$ in this range.
	Consider the vector
	\begin{equation}
		x^{\dagger} \Pi_{I}(|x^{\dagger}|) \ket{\Omega}.
	\end{equation}
	Using the expression $x^{\dagger} = |x| u^{\dagger},$ and the easy-to-verify identity $u |x| u^{\dagger} = |x^{\dagger}|,$ we can rewrite this vector as
	\begin{equation}
		x^{\dagger} \Pi_{I}(|x^{\dagger}|) \ket{\Omega}
			= \Pi_{I}(|x|) x^{\dagger} \ket{\Omega}.
	\end{equation}
	The norm-squared of this vector can be written as
	\begin{equation}
		\lVert x^{\dagger} \Pi_{I}(|x^{\dagger}|) \ket{\Omega} \rVert^2
		= \langle x^\dagger \Omega | \Pi_{I}(|x|) x^{\dagger} \Omega \rangle.
	\end{equation}
	Using the fundamental identity \eqref{eq:exact-K-swap} for the modular operator, we may rewrite this as
	\begin{equation} \label{eq:resolvent-lemma-intermediate}
		\lVert x^{\dagger} \Pi_{I}(|x^{\dagger}|) \ket{\Omega} \rVert^2
			= \langle x \Pi_{I}(|x|) \Omega | \Delta x \Omega \rangle
			= \langle \Pi_{I}(|x^{\dagger}|) x \Omega | \Delta x \Omega \rangle.
	\end{equation}

	By contrast, consider the action of the operator $x'$ on the vector $\Pi_{I}(|x^{\dagger}|) \ket{\Omega}.$
	We have
	\begin{equation}
		x' \Pi_{I}(|x^{\dagger}|) \ket{\Omega}
			= \Pi_{I}(|x^{\dagger}|) x' \ket{\Omega}
			= \Pi_{I}(|x^{\dagger}|) (z - \Delta) x \ket{\Omega}.
	\end{equation}
	Taking the norm squared and expanding in terms of the inner product, we obtain
	\begin{equation}
		\lVert x' \Pi_{I}(|x^{\dagger}|) \ket{\Omega} \rVert^2
			= |z|^2 \lVert \Pi_{I}(|x^{\dagger}|) x \ket{\Omega} \rVert^2
				+ \lVert \Pi_{I}(|x^{\dagger}|) \Delta x \ket{\Omega} \rVert^2
				- 2 \Re\left( z \langle \Delta x \Omega | \Pi_{I}(|x^{\dagger}|) x \Omega \rangle \right).
	\end{equation}
	The last term in this expression can be compared to equation \eqref{eq:resolvent-lemma-intermediate} to write
	\begin{equation}
		\lVert x' \Pi_{I}(|x^{\dagger}|) \ket{\Omega} \rVert^2
		= |z|^2 \lVert \Pi_{I}(|x^{\dagger}|) x \ket{\Omega} \rVert^2
		+ \lVert \Pi_{I}(|x^{\dagger}|) \Delta x \ket{\Omega} \rVert^2
		- 2 \Re(z) \lVert x^{\dagger} \Pi_{I}(|x^{\dagger}|) \ket{\Omega} \rVert^2.
	\end{equation}
	To make the first two terms look more like equation \eqref{eq:resolvent-lemma-intermediate}, we can use the universal inequality $p^2 + |z|^2 q^2 \geq 2 |z| p q,$ which follows from the expression $(p - |z| q)^2 \geq 0.$
	Using this gives
	\begin{equation}
		\lVert x' \Pi_{I}(|x^{\dagger}|) \ket{\Omega} \rVert^2
		\geq 2 |z| \lVert \Pi_{I}(|x^{\dagger}|) x \ket{\Omega} \rVert \lVert \Pi_{I}(|x^{\dagger}|) \Delta x \ket{\Omega} \rVert
		- 2 \Re(z) \lVert x^{\dagger} \Pi_{I}(|x^{\dagger}|) \ket{\Omega} \rVert^2.
	\end{equation}
	Applying the Cauchy-Schwarz inequality then gives
	\begin{equation}
		\lVert x' \Pi_{I}(|x^{\dagger}|) \ket{\Omega} \rVert^2
		\geq 2 |z| \left|\langle \Pi_{I}(|x^{\dagger}|) x \Omega | \Delta x \Omega \rangle \right|
		- 2 \Re(z) \lVert x^{\dagger} \Pi_{I}(|x^{\dagger}|) \ket{\Omega} \rVert^2.
	\end{equation}
	Finally, invoking equation \eqref{eq:resolvent-lemma-intermediate} gives
	\begin{equation}
		\lVert x' \Pi_{I}(|x^{\dagger}|) \ket{\Omega} \rVert^2
		\geq 2 (|z| - \Re(z)) \lVert x^{\dagger} \Pi_{I}(|x^{\dagger}|) \ket{\Omega} \rVert^2,
	\end{equation}
	and hence
	\begin{equation}
		\lVert x^{\dagger} \Pi_{I}(|x^{\dagger}|) \ket{\Omega} \rVert
			\leq \frac{\lVert x' \rVert_{\infty}}{\sqrt{2 (|z| - \Re(z))}}
				\lVert \Pi_{I}(|x^{\dagger}|) \ket{\Omega} \rVert.
	\end{equation}
	Now let us denote the endpoints of the interval $I$ by $I = [s_1, s_2].$
	We have
	\begin{equation}
		\lVert x^{\dagger} \Pi_{I}(|x^{\dagger}|) \ket{\Omega} \rVert
			\geq s_1	\lVert \Pi_{I}(|x^{\dagger}|) \ket{\Omega} \rVert.
	\end{equation}
	Hence 
	\begin{equation}
		\lVert \Pi_{I}(|x^{\dagger}|) \ket{\Omega} \rVert
		\leq \frac{1}{s_1} \frac{\lVert x' \rVert_{\infty}}{\sqrt{2 (|z| - \Re(z))}}
		\lVert \Pi_{I}(|x^{\dagger}|) \ket{\Omega} \rVert.
	\end{equation}
	If the norm appearing in this inequality is nonzero, then we can divide by it on either side, and therefore obtain the inequality
	\begin{equation}
		s_1 \leq \frac{\lVert x' \rVert_{\infty}}{\sqrt{2 (|z| - \Re(z))}}.
	\end{equation}
	So if, by contrast, $I$ is such that its left endpoint $s_1$ satisfies
	\begin{equation}
		s_1 > \frac{\lVert x' \rVert_{\infty}}{\sqrt{2 (|z| - \Re(z))}},
	\end{equation}
	then we must have
	\begin{equation}
		x^{\dagger} \Pi_{I}(|x^{\dagger}|) \ket{\Omega} = 0,
	\end{equation}
	which, by the fact that $\ket{\Omega}$ is separating, implies
		\begin{equation}
		x^{\dagger} \Pi_{I}(|x^{\dagger}|) = 0,
	\end{equation}
	and hence $\Pi_{I}(|x^{\dagger}|) = 0.$
	
	We conclude that the spectral support of $|x^{\dagger}|$ lies entirely within the interval lower bounded by zero and upper bounded by $\frac{\lVert x' \rVert_{\infty}}{\sqrt{2 (|z| - \Re(z))}}$, and therefore that $x^{\dagger}$ is bounded and we have
	\begin{equation}
		\lVert x \rVert_{\infty} = \lVert x^{\dagger} \rVert_{\infty}	
			\leq \frac{\lVert x' \rVert_{\infty}}{\sqrt{2 (|z| - \Re(z))}},
	\end{equation}
	as desired.
\end{proof}

\subsection{Constructing the tidy subspace}

\label{sec:tidy-subspace}

We will now use Takesaki's resolvent lemma, proved in the previous subsection, to study vectors like $f(\Delta) a \ket{\Omega}$ for certain functions $f$ analytic in a neighborhood of $[0, \infty).$
The idea will be to restrict our attention to functions $f$ that die off sufficiently quickly at infinity, and then to show that for $a' \in \A',$ we can write $f(\Delta) a' \ket{\Omega}$ as a contour integral
\begin{equation}
	f(\Delta) a' \ket{\Omega}
		= \frac{1}{2\pi i} \int_{\gamma} dz\, f(z) (z-\Delta)^{-1} a' \ket{\Omega},
\end{equation}
where the integral has the properties of the Bochner integral on Hilbert space discussed in section \ref{sec:analytic-operators}.
Using Takesaki's resolvent lemma (lemma \ref{lem:resolvent-lemma}), we will write $(z - \Delta)^{-1} a' \ket{\Omega}$ as
\begin{equation}
	(z - \Delta)^{-1} a' \ket{\Omega}
		= a_z \ket{\Omega}
\end{equation}
for some operator $a_z \in \A.$
This will let us express the contour integral as
\begin{equation}
	f(\Delta) a' \ket{\Omega}
		= \frac{1}{2\pi i} \int_{\gamma} dz\, f(z) a_z \ket{\Omega}.
\end{equation}
We will then show that this can be written in terms of an operator in $\A$ as $a_f \ket{\Omega}.$
By judiciously choosing a sequence of analytic functions $f_k$ to approximate the Heaviside function, we will be able to construct an operator $a_{[0, \lambda_2]}$ satisfying the equation
\begin{equation} \label{eq:a'-to-a-projection}
	\Theta(\lambda_2 - \Delta) a' \ket{\Omega} = a_{[0, \lambda_2]} \ket{\Omega}.
\end{equation}
By invoking a symmetric argument, with the substitutions $\A \leftrightarrow \A'$ and $\Delta \leftrightarrow \Delta^{-1},$ we will show that for any $a \in \A,$ there exists an operator $a'_{[\lambda_1, \infty)}$ in $\A'$ satisfying
\begin{equation} \label{eq:a-to-a'-projection}
	\Theta(\Delta - \lambda_1) a \ket{\Omega} = 
		a'_{[\lambda_1, \infty)} \ket{\Omega}.
\end{equation}

By combining equations \eqref{eq:a'-to-a-projection} and \eqref{eq:a-to-a'-projection}, we will be able to show, for any $a \in \A$, the existence of operators $a_{[\lambda_1, \lambda_2]} \in \A$ and $a'_{[\lambda_1, \lambda_2]}$ in $\A'$ satisfying
\begin{equation}
	\Theta(\lambda_2 - \Delta) \Theta(\Delta - \lambda_1) a \ket{\Omega}
	= a_{[\lambda_1, \lambda_2]} \ket{\Omega}
	= a'_{[\lambda_1, \lambda_2]} \ket{\Omega}.
\end{equation}
We will use this to construct a dense set of states in $\H$ that have support in compact spectral ranges of $\Delta$ and $\Delta^{-1}$, and that can be written either in terms of an operator in $\A$ acting on $\ket{\Omega}$ or in terms of an operator in $\A'$ acting on $\ket{\Omega}.$\footnote{This is interesting in part because while $\A \ket{\Omega}$ and $\A' \ket{\Omega}$ are both dense in $\H$, we did not know a priori whether their intersection was dense --- the vectors constructed here furnish a dense subset of the intersection $\A \ket{\Omega} \cap \A' \ket{\Omega}.$}
These vectors will even have the property that if we act on them with an integer power of the modular operator, $\Delta^n,$ they can still be written in terms of operators in $\A$ or $\A',$ and that the norm of the resulting operator is bounded by some exponential function of $n$.

The operators in $\A$ that produce the special states described in the preceding paragraph will be called the \textbf{tidy operators} in $\A.$
The ones in $\A'$ will be called the tidy operators in $\A'.$
These can be thought of as the operators for which the modular operator ``looks bounded.''

The rest of the section is written as a series of lemmas, propositions, and theorems.

\begin{lemma} \label{lem:full-unbounded-residue}
	Let $f$ be a bounded, analytic function in a neighborhood of $[0, \infty),$ and let $\gamma$ be a simple contour in that neighborhood surrounding $[0, \infty)$ counterclockwise, and such that $\gamma$ is contained in some bounded horizontal strip (so that the top and bottom parts of the contour do not get arbitrarily far away from each other at infinity).
	Suppose further that in the interior of the contour, $f(z)$ vanishes uniformly in the limit $\text{Re}(z) \to \infty$, and does so quickly enough that the real integral
	\begin{equation}
		\int_{\gamma} ds |f(z)| \lVert (z - \Delta)^{-1} \rVert_{\infty}
	\end{equation}
	is finite, where $s$ is the arclength parameter.
	Then for any vector $\ket{\psi},$ we have
	\begin{equation} \label{eq:full-unbounded-residue}
		f(\Delta) \ket{\psi}
			= \frac{1}{2\pi i} \int dz\, f(z) (z - \Delta)^{-1} \ket{\psi},
	\end{equation}
	where this integral is evaluated in the sense of the Bochner integral from section \ref{sec:analytic-operators}.
\end{lemma}
\begin{proof}
	Let $\Pi_n$ be the spectral projection of $\Delta$ onto the range $[0, n],$ and let $\Delta_n = \Delta \Pi_n$ be the restriction of $\Delta$ to that spectral range.
	For each $n,$ we can construct a vertical segment $v_n$ crossing the real line at $n+1/2$ and cutting through the contour $\gamma$; see figure \ref{fig:contour-cutoff}.
	We call the part of $\gamma$ to the left of this vertical segment $\gamma_n.$
	
	\begin{figure}
		\centering
		\includegraphics{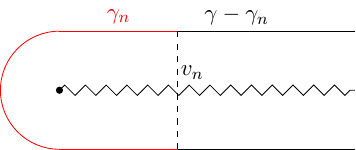}
		\caption{Given a sufficiently nice contour $\gamma$ enclosing the positive real axis $[0, \infty]$, and given an integer $n,$ one can draw a vertical segment $v_n$ passing through the real axis at $n + 1/2,$ which splits the contour into a piece $\gamma_n$ to the left of the vertical segment, and a piece $\gamma - \gamma_n$ to the right of the vertical segment.}
		\label{fig:contour-cutoff}
	\end{figure}
	
	It is a straightforward exercise using the spectral theorem to show the identity $f(\Delta) \Pi_n = f(\Delta_n) \Pi_n.$
	From this, and from the fact that each $\Delta_n$ is bounded, we may apply the residue formula \eqref{eq:bounded-residue} to obtain
	\begin{equation} \label{eq:projected-unbounded-residue}
		f(\Delta) \Pi_n \ket{\psi}
			= f(\Delta_n) \Pi_n \ket{\psi}
			= \frac{1}{2 \pi i} \int_{\gamma_n + v_n} dz\, f(z) (z - \Delta_n)^{-1} \Pi_n \ket{\psi}.
	\end{equation}
	Since the spectrum of $\Delta_n$ is contained in the range $[0, n],$ we may actually deform the contour to $\gamma_{m} + v_m$ for any $m \geq n.$
	It is easy to see from the assumptions of the lemma that the integral over $v_m$ vanishes in the limit $m \to \infty,$ so we may write
	\begin{align}
	\begin{split}
		f(\Delta) \Pi_n \ket{\psi}
			& = \frac{1}{2 \pi i} \int_{\gamma} dz\, f(z) (z - \Delta_n)^{-1} \Pi_n \ket{\psi} \\
			& = \frac{1}{2 \pi i} \int_{\gamma} dz\, f(z) (z - \Delta)^{-1} \Pi_n \ket{\psi}.
	\end{split}
	\end{align}
	The spectral theorem implies that in the limit $n \to \infty,$ the projections $\Pi_n$ converge strongly to the identity operator.
	So taking the $n \to \infty$ limit on the left side of this equation gives the vector $f(\Delta) \ket{\psi}.$
	Taking the $n \to \infty$ limit on the right side proves the theorem provided that we can interchange the limit and the integral.
	This interchange can be justified using a straightforward application of Lebesgue's dominated convergence theorem.
\end{proof}

\begin{prop} \label{prop:analytic-mollifier}
	Let $f$ be a bounded analytic function in a neighborhood of $[0, \infty)$ and $\gamma$ a contour satisfying the conditions of the previous lemma.
	Suppose also that the function $z^{1/2} f(z)$ is bounded on $[0, \infty).$
	Fix $a' \in \A'.$
	The vector
	\begin{equation}
		f(\Delta) a' \ket{\Omega}
	\end{equation}
	can be written as $a_f \ket{\Omega}$ for some unique $a_f \in \A.$
	Furthermore, the norm of this operator is bounded by
	\begin{align}
		\begin{split}
			\lVert a_f \rVert_{\infty}
			& \leq \frac{\lVert a' \rVert_{\infty}}{2 \pi} \int_{\gamma} ds\, \frac{|f(z)|}{\sqrt{2 (|z| - \Re(z))}}.
		\end{split}
	\end{align}
\end{prop}
\begin{proof}
	Since the function $z^{1/2} f(z)$ is bounded by assumption, the vector $f(\Delta) a' \ket{\Omega}$ is in the domain of $\Delta^{1/2}.$
	As discussed in appendix \ref{app:tomita}, this means there exists some operator $a_f$ affiliated to $\A$ satisfying $a_f \ket{\Omega} = f(\Delta) a' \ket{\Omega}.$
	Our job is to show that $a_f$ is bounded.
	Uniqueness is then an immediate consequence of the fact that $\ket{\Omega}$ is separating.
	
	Given $b' \in \A',$ by the previous lemma and the properties of the Bochner integral discussed in section \ref{sec:analytic-operators}, we have
	\begin{align}
		\begin{split}
			a_f b' \ket{\Omega}
				& = b' a_f \ket{\Omega} \\
				& = b' f(\Delta) a' \ket{\Omega} \\
				& = \frac{1}{2 \pi i} b' \int_{\gamma} dz\, f(z) (z - \Delta)^{-1} a' \ket{\Omega} \\
				& = \frac{1}{2 \pi i} \int_{\gamma} dz\, f(z) b' (z - \Delta)^{-1} a' \ket{\Omega}.
		\end{split}
	\end{align}
	By Takesaki's resolvent lemma (lemma \ref{lem:resolvent-lemma}), there exists some bounded operator $a_z \in \A$ satisfying $(z-\Delta)^{-1} a' \ket{\Omega} = a_z \ket{\Omega}.$
	Using this identity gives
	\begin{align}
	\begin{split}
		a_f b' \ket{\Omega}
			& = \frac{1}{2 \pi i} \int_{\gamma} dz\, f(z) b' a_z \ket{\Omega} \\
			& = \frac{1}{2 \pi i} \int_{\gamma} dz\, f(z) a_z b' \ket{\Omega}.
	\end{split}
	\end{align}
	The norm of the vector on the left-hand side satisfies the bound
	\begin{align}
	\begin{split}
		\lVert a_f b' \ket{\Omega} \rVert
			& =  \frac{1}{2 \pi } \left \lVert \int_{\gamma} dz\, f(z) a_z b' \ket{\Omega} \right \rVert \\
			& \leq \frac{1}{2 \pi} \lVert b' \ket{\Omega} \rVert \int_{\gamma} ds\, |f(z)| \lVert a_z \rVert
	\end{split}
	\end{align}
	Applying the specific bound derived in lemma \ref{lem:resolvent-lemma} gives the estimate
	\begin{align}
	\begin{split}
		\lVert a_f b' \ket{\Omega} \rVert
		& \leq \frac{\lVert a' \rVert_{\infty} }{2 \pi} \lVert b' \ket{\Omega} \rVert \int_{\gamma} ds\, \frac{|f(z)|}{\sqrt{2 (|z| - \Re(z))}}.
	\end{split}
	\end{align}
	The assumptions of lemma \ref{lem:full-unbounded-residue} guarantee that the integral is finite, so the action of $a_f$ on $b' \ket{\Omega}$ is bounded by a constant times $\lVert b' \ket{\Omega} \rVert.$
	As explained in appendix \ref{app:tomita}, the vectors $b' \ket{\Omega}$ form a core\footnote{The core of an operator was defined in definition \ref{def:core}.} for affiliated operators constructed from the domain of $\Delta^{1/2},$ so showing that $a_f$ is bounded on these vectors suffices to show that it is bounded on all vectors.
\end{proof}

\begin{theorem} \label{thm:brutal-mollifier}
	Let $\Theta$ be the Heaviside theta function, with the convention $\Theta(0) = \frac{1}{2}.$
		
	For $\lambda_2 > 0$ and $a' \in \A,$ there exists a unique $a_{[0, \lambda_2]} \in \A$ satisfying
	\begin{equation}
		a_{[0, \lambda_2]} \ket{\Omega} = \Theta(\lambda_2 - \Delta) a' \ket{\Omega}.
	\end{equation}
	For any nonnegative integer $n,$ there also exists a unique $a_{[0, \lambda_2], n} \in \A$ satisfying
	\begin{equation}
		a_{[0, \lambda_2], n} \ket{\Omega} = \Delta^n \Theta(\lambda_2 - \Delta) a' \ket{\Omega}.
	\end{equation}

	Furthermore, the norm of these operators is exponentially bounded in $n$, in that there exist $\lambda_2$-independent constants $\alpha, \beta$ satisfying
	\begin{equation}
		\lVert a_{[0, \lambda_2], n}\rVert_{\infty}
			\leq \alpha e^{\beta n}.
	\end{equation}
\end{theorem}
\begin{proof}
	Note that because the step function $\Theta(\lambda_2 - \Delta)$ cuts off the large spectral subspaces of $\Delta,$ each vector
	\begin{equation}
		\Delta^n \Theta(\lambda_2 - \Delta) a' \ket{\Omega}
	\end{equation}
	is in the domain of $\Delta^{1/2},$ and hence (by the properties of the Tomita operator explained in appendix \ref{app:tomita}) can be written uniquely as 
	\begin{equation}
		a_{[0, \lambda_2], n} \ket{\Omega} = \Delta^n \Theta(\lambda_2 - \Delta) a' \ket{\Omega},
	\end{equation}
	where $a_{[0, \lambda_2], n}$ is a potentially unbounded operator affiliated with $\A$.
	We want to show that it is bounded and give an explicit bound on its norm as a function of $n$.
		
	This is an engineering problem.
	What we need is a sequence of functions $f_j,$ each analytic in a neighborhood of $[0, \infty),$ and satisfying the conditions of proposition \ref{prop:analytic-mollifier}.
	We also want these functions to decay sufficiently quickly at infinity so that each for each nonnegative integer $n,$ the functions $z^n f_j(z)$ also satisfy the conditions of proposition \ref{prop:analytic-mollifier}. 
	Finally, we want this sequence of functions to approximate the step function $\Theta(\lambda_2 - \Delta)$ in the sense that we have
	\begin{equation}
		\lim_{j \to \infty} \Delta^n f_j(\Delta) a' \ket{\Omega}
			= \Delta^n \Theta(\lambda_2 - \Delta) a' \ket{\Omega}.
	\end{equation}
	This will allow us to write the following chain of equalities for any $b' \in \A'$.
	\begin{align}
		\begin{split}
			a_{[0, \lambda_2], n} b' \ket{\Omega}
				& = b' a_{[0, \lambda_2], n} \ket{\Omega} \\
				& = b' \Delta^n  \Theta(\lambda_2 - \Delta) a' \ket{\Omega} \\
				& = b' \lim_{j \to \infty} \Delta^n f_j(\Delta) a' \ket{\Omega} \\
				& = b' \lim_{j \to \infty} a_{f_j, n} \ket{\Omega} \\
				& = \lim_{j \to \infty} a_{f_j, n} b' \ket{\Omega}.
		\end{split}
	\end{align}
	This lets us write the norm as
	\begin{align}
		\begin{split}
		\lVert a_{[0, \lambda_2], n} b' \ket{\Omega} \rVert
			& = \lim_{j \to \infty} \lVert a_{f_j, n} b' \ket{\Omega} \rVert \\
			& \leq \lVert b' \ket{\Omega} \rVert \limsup_{j \to \infty} \lVert a_{f_j, n} \rVert_{\infty}.
		\end{split}
	\end{align}
	So long as our functions $f_j$ are sufficiently well behaved that this limsup is bounded, we may conclude that $a_{[0, \lambda_2], n}$ is a bounded operator.
	But by the estimate on the operator norms given in proposition \ref{prop:analytic-mollifier}, it suffices to check
	\begin{equation} \label{eq:limsup}
		\limsup_{j \to \infty} \int_{\gamma} ds\, \frac{|z|^n |f_j(z)|}{\sqrt{2 (|z| - \Re(z))}} < \infty.
	\end{equation}

	The classic sequence of analytic functions approximating a step function is the sequence of sigmoid functions
	\begin{equation}
		f_j(z) = \frac{1}{1 + e^{j (z - \lambda_2)} }.
	\end{equation}
	Indeed, it is a simple exercise using the spectral theorem to show that $\Delta^n f_j(\Delta)$ converges strongly to $\Delta^n \Theta(\lambda_2 - \Delta).$\footnote{To see this concretely, fix $\ket{\psi} \in \H,$ write
	\begin{equation}
		\lVert\Delta^n \Theta(\lambda_2 - \Delta) \ket{\psi} - \Delta^n f_j(\Delta) \ket{\psi} \rVert^2
			= \int d\mu_{\psi, \psi}(t)\, |t^n \Theta(\lambda_2 - t) - t^n f_j(t)|^2,
	\end{equation}
	and apply the dominated convergence theorem to show that this goes to zero in the limit $j \to \infty.$}
	So all we need to do is pick a good contour $\gamma$ on which we can estimate inequality \eqref{eq:limsup}.
	A good contour is provided by combining the half-lines $t + 2 \pi i$ and $t - 2 \pi i$ for $t \geq 0$ with the half-circle of radius $2 \pi$ at the origin.
	See figure \ref{fig:specific-contour}.
	
	\begin{figure}
		\centering
		\includegraphics{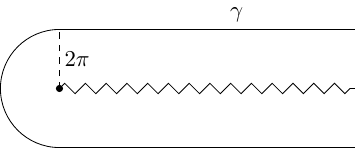}
		\caption{The specific contour used to evaluate the integral in equation \eqref{eq:limsup}, to show that operators in the tidy subspace satisfy a nice exponential bound.}
		\label{fig:specific-contour}
	\end{figure}
	
	For the portion of the contour consisting of infinite half-lines, the function $f_j(z)$ is equal to $\frac{1}{1 + e^{j (t - \lambda_2)}}.$
	For fixed $n$ and large $j,$ it converges to a step function in the range $[0, \lambda_2]$.
	Consequently, the contribution of the half-lines to the contour integral in equation \eqref{eq:limsup} gives\footnote{I have been a little cavalier here about moving the limit $\limsup_{j \to \infty}$ inside the integral, but this can be justified rigorously by a simple application of the dominated convergence theorem.}
	\begin{equation}
		\limsup_{j \to \infty} \int_{\text{half-line}} ds\, \frac{|z|^n |f_j(z)|}{\sqrt{2 (|z| - \Re(z))}}
			\leq \int_{0}^{\lambda_2} dt\, \frac{(t^2 + 4 \pi^2)^{n/2}}{\sqrt{2 (\sqrt{t^2 + 4 \pi^2} - t)}}.
	\end{equation}
	Under a fairly brutal approximation, we can bound this by
	\begin{equation}
		\limsup_{j \to \infty} \int_{\text{half-line}} ds\, \frac{|z|^n |f_j(z)|}{\sqrt{2 (|z| - \Re(z))}}
		\leq \lambda_2 \frac{(\lambda_2^2 + 4 \pi^2)^{n/2}}{\sqrt{2(\sqrt{\lambda_2^2 + 4 \pi^2} - \lambda_2)}}.
	\end{equation}
	Using another fairly brutal approximation, it is straightforward to see that the half-circle portion of the contour integral is bounded by
	\begin{equation}
		\limsup_{j \to \infty} \int_{\text{half-circle}} ds\, \frac{|z|^n |f_j(z)|}{\sqrt{2 (|z| - \Re(z))}}
			\leq (2 \pi)^{n+1} \frac{\sqrt{\pi}}{2}.
	\end{equation}
	So the full contour integral is bounded by
	\begin{align}
	\begin{split}
			& \limsup_{j \to \infty} \int_{\gamma} ds\, \frac{|z|^n |f_j(z)|}{\sqrt{2 (|z| - \Re(z))}} \\
			& \qquad \qquad = 2 \limsup_{j \to \infty} \int_{\text{half-line}} ds\, \frac{|z|^n |f_j(z)|}{\sqrt{2 (|z| - \Re(z))}} 
				+ \limsup_{j \to \infty} \int_{\text{half-circle}} ds\, \frac{|z|^n |f_j(z)|}{\sqrt{2 (|z| - \Re(z))}} \\
			&\qquad \qquad \leq 2 \lambda_2 \frac{(\lambda_2^2 + 4 \pi^2)^{n/2}}{\sqrt{2(\sqrt{\lambda_2^2 + 4 \pi^2} - \lambda_2)}}
				+ (2 \pi)^{n+1} \frac{\sqrt{\pi}}{2}.
	\end{split}
	\end{align}
	Since the norm $a_{[0, \lambda_2], n}$ is bounded by $\frac{\lVert a' \rVert_{\infty}}{2 \pi}$ times this integral, it is easy to see that there exist some constants $\alpha, \beta$ satisfying
	\begin{equation}
		\lVert a_{[0, \lambda_2], n}\rVert_{\infty}
			\leq \alpha e^{\beta n}.
	\end{equation}
\end{proof}

\begin{remark} \label{rem:symmetric-mollifier}
	By a completely symmetric argument to those given above, substituting $\Delta \leftrightarrow \Delta^{-1}$ and $\A \leftrightarrow \A'$, it is easy to see that for $\lambda_1 > 0,$ if we start with an operator $a \in \A,$ then there exists, for every integer $n \leq 0,$ an operator $a'_{[\lambda_1, \infty),n} \in \A'$ satisfying
	\begin{equation}
		a'_{[\lambda_1, \infty), n} \ket{\Omega}
			= \Delta^n \Theta(\Delta - \lambda_1) a \ket{\Omega},
	\end{equation}
	and that these operators are exponentially bounded in $|n|$.
\end{remark}

\begin{theorem}[Construction of the tidy subspace] \label{thm:tidy-construction}
	Fix $\lambda_1, \lambda_2$ satisfying $0 < \lambda_1 < \lambda_2.$
	Fix $a \in \A.$
	Then for any integer $n,$ there exist unique operators
	\begin{align}
		a_{[\lambda_1, \lambda_2], n} 
			& \in \A, \\
		a'_{[\lambda_1, \lambda_2], n}
			& \in \A'
	\end{align}
	satisfying
	\begin{equation}
		\Delta^n \Theta(\lambda_2 - \Delta) \Theta(\Delta - \lambda_1) a \ket{\Omega}
			= a_{[\lambda_1, \lambda_2], n} \ket{\Omega}
			= a'_{[\lambda_1, \lambda_2], n} \ket{\Omega}.
	\end{equation}
	Furthermore, these operators are exponentially bounded in $|n|.$
	
	The set of all operators of the form $a_{[\lambda_1, \lambda_2], n}$ is called the \textbf{tidy subspace} $\A_{\text{tidy}}.$
	The vectors $\A_{\text{tidy}} \ket{\Omega}$ are dense in $\H$.
\end{theorem}
\begin{proof}
	If we have $n \geq 0,$ then we write
	\begin{equation}
		\Delta^n \Theta(\lambda_2 - \Delta) \Theta(\Delta - \lambda_1)  a \ket{\Omega}
			= (\Delta^n \Theta(\lambda_2 - \Delta)) \Theta(\Delta - \lambda_1) a \ket{\Omega}.
	\end{equation}
	If we have $n < 0,$ then we write
	\begin{equation}
		\Delta^n \Theta(\lambda_2 - \Delta) \Theta(\Delta - \lambda_1) a \ket{\Omega}
			= \Theta(\lambda_2 - \Delta) (\Delta^n \Theta(\Delta - \lambda_1)) a \ket{\Omega}
	\end{equation}
	Applying remark \ref{rem:symmetric-mollifier} and then theorem \ref{thm:brutal-mollifier} proves the first part of the theorem.
	
	Density of $\A_{\text{tidy}} \ket{\Omega}$ in $\A\ket{\Omega}$ follows from the limit
	\begin{equation}
		a \ket{\Omega}
			= \lim_{n \to \infty} \Theta(n - \Delta) \Theta\left(\Delta - \frac{1}{n} \right) a \ket{\Omega},
	\end{equation}
	which is easy to show using the spectral theorem and invertibility of $\Delta.$
\end{proof}

\begin{remark}
	The last theorem of this section tells us that $\A_{\text{tidy}} \ket{\Omega}$ is a core for any real power of the modular operator.
	(See definition \ref{def:core} for the definition of a core.)
	This will be valuable in the next subsection for two reasons.
	First, we will use it to apply theorem \ref{thm:operator-continuation}, which lets us constrain operators of the form $\Delta^{n} a \Delta^{-n}$ in terms of their action on a core of $\Delta^{-n}.$
	Second, we will use it to show the identity $\A_{\text{tidy}}'' = \A,$ which will allow us to obtain the general case of Tomita's theorem after proving that the theorem holds for tidy operators.
\end{remark}

\begin{prop} \label{prop:core}
	For any real number  $x,$ the space $\A_{\text{tidy}} \Omega$ is a core for $\Delta^{x}.$
\end{prop}
\begin{proof}
	We will use the characterization of a core given in remark \ref{rem:core-graph}.
	Namely, we will show that $\A_{\text{tidy}} \ket{\Omega}$ is a core for $\Delta^x$ by showing that vectors of the form $a \ket{\Omega} \oplus \Delta^x a \ket{\Omega},$ with $a \in \A_{\text{tidy}},$ are dense in the graph of $\Delta^x.$
	
	Suppose that $\ket{\psi}$ is a vector in the domain of $\Delta^{x}$ such that $\psi \oplus \Delta^{x} \psi$ is orthogonal to all such vectors.
	I.e., suppose that for all $a \in \A_{\text{tidy}},$ we have
	\begin{align}
		\begin{split}
			0
				& = (\langle a \Omega | \oplus \langle \Delta^x a \Omega |) (|\psi\rangle \oplus \Delta^x |\psi \rangle) \\
				& = \langle a \Omega | \psi \rangle + \langle \Delta^{2x} a \Omega | \psi \rangle \\
				& = \langle (1 + \Delta^{2x}) a \Omega | \psi \rangle.
		\end{split}
	\end{align}
	Our goal is to show that whenever this expression is satisfied, the vector $\ket{\psi}$ vanishes.
	It suffices to show that the space $(1 + \Delta^{2x}) \A_{\text{tidy}} \ket{\Omega}$ is dense in Hilbert space.
	
	To see this, note that per the construction of $\A_{\text{tidy}}$ from theorem \ref{thm:tidy-construction}, each $a \in \A_{\text{tidy}}$ satisfies an equation like
	\begin{equation}
		a \Omega
		= \Theta(\lambda_2 - \Delta) \Theta(\Delta - \lambda_1) b \Omega
	\end{equation}
	for some $0 < \lambda_1 < \lambda_2$ and some $b \in \A.$
	The operator
	\begin{equation}
		(1 + \Delta^{2x}) \Theta(\lambda_2 - \Delta) \Theta(\Delta - \lambda_1)
	\end{equation}
	is bounded, injective, and invertible on the spectral subspace of $\Delta$ corresponding to the range $[\lambda_1, \lambda_2].$
	It therefore maps any dense subset of $\H$ into a dense subspace of that spectral subspace; in particular, the space
	\begin{equation}
		(1 + \Delta^{2x}) \Theta(\lambda_2 - \Delta) \Theta(\Delta - \lambda_1) \mathcal{A} \ket{\Omega}
	\end{equation}
	is dense in the spectral subspace of $\Delta$ for the range $[\lambda_1, \lambda_2]$ by the assumption that $\mathcal{A} \ket{\Omega}$ is dense in $\H$.
	Taking $\lambda_1 \to 0$ and $\lambda_2 \to \infty$ shows that $(1 + \Delta^{2x})\A_{\text{tidy}} \ket{\Omega}$ is dense in the spectral subspace of $\Delta$ for the range $(0, \infty)$, which by invertibility of $\Delta$ is equal to all of $\H.$
\end{proof}

\subsection{Finishing the proof: analytic extensions of modular flow}
\label{sec:main-proof}

\begin{remark}
	\label{rem:tidy-properties}
	We will now forget all of the details of the construction of the tidy subspace from the previous subsection, and use only the following facts.
	\begin{itemize}
		\item There is a subspace $\A_{\text{tidy}}$ of $\A$ such that $\A_{\text{tidy}} \ket{\Omega}$ is dense in $\H,$ and is a core for each operator $\Delta^x.$
		\item Each operator $a \in \A_{\text{tidy}}$ has the property that $a \ket{\Omega}$ is supported in some spectral subspace of $\Delta$ that is a closed interval in $(0, \infty).$
		\item Each vector $a \ket{\Omega}$ can be written as $a' \ket{\Omega}$ for some $a' \in \A'.$
		\item Each vector $\Delta^n a \ket{\Omega}$ for integer $n$ can be written as $a_n \ket{\Omega} = a'_n \ket{\Omega}$ for some $a_n \in \A,$ $a'_n \in \A.$
		There exist constants $\alpha, \beta$ satisfying
		\begin{align}
			\lVert a_n \rVert_{\infty}
				& \leq \alpha e^{\beta |n|}.
		\end{align}
	\end{itemize}

	There is a piece of notation that was introduced in this list which we will use frequently in what follows, so we reemphasize it in the following definition.
\end{remark}
\begin{definition}
	For an operator $a$ in the tidy subspace, we will denote by $a'$ the operator in $\A'$ satisfying $a \ket{\Omega} = a' \ket{\Omega}.$
	We will denote by $a_n$ and $a'_n$ the operators in $\A$ and $\A'$ satisfying
	\begin{equation}
		\Delta^n a \ket{\Omega} = a_n \ket{\Omega} = a_n'\ket{\Omega}
	\end{equation}
\end{definition}

\begin{remark}
	We will now proceed to show that for $a$ in the tidy subspace, the function $it \mapsto \Delta^{-it} a \Delta^{it}$ admits an entire analytic continuation that is exponentially bounded at infinity.
	We will need one lemma, which tells us how to think of the operators $a^{\dagger}$ for $a \in \A_{\text{tidy}}.$
	We will then show the analytic continuations of modular flow exist, and reach the conclusion of Tomita's theorem by applying Carlson's theorem.
\end{remark}

\begin{lemma}[Dagger-Ladder Lemma] \label{lem:dagger-ladder}
	Let $a$ be in the tidy subspace $\A_{\text{tidy}}$.
	Then we have
	\begin{equation}
		a_n^{\dagger} \ket{\Omega} = (a_{n+1}')^{\dagger} \ket{\Omega}.
	\end{equation}
\end{lemma}
\begin{proof}
	Fix arbitrary $b \in \A$, and write
	\begin{align}
	\begin{split}
		\langle (a_{n+1}')^{\dagger} \Omega | b \Omega \rangle
			& = \langle b^{\dagger} \Omega | a_{n+1}' \Omega \rangle \\
			&= \langle b^{\dagger} \Omega | a_{n+1} \Omega \rangle \\
			& = \langle b^{\dagger} \Omega | \Delta a_n \Omega \rangle \\
			& = \langle a_n^{\dagger} \Omega | b \Omega \rangle.
	\end{split}
	\end{align}
	Since vectors of the form $b |\Omega \rangle$ are dense in $\H,$ this proves the lemma.
\end{proof}

\begin{theorem}
	If $a$ is in the tidy subspace $\A_{\text{tidy}},$ then for integer values of $n,$ the operator $\Delta^{-n} a \Delta^{n}$ is densely defined and bounded on its domain, with closure $a_{-n}.$
\end{theorem}
\begin{proof}
	First note that for any $b$ in the tidy subspace, the vector $b \ket{\Omega}$ is in the domain of $\Delta^n.$
	If we can show that $a \Delta^n b \ket{\Omega}$ is in the domain of $\Delta^{-n},$ then we will have proved that $b \ket{\Omega}$ is in the domain of $\Delta^{-n} a \Delta^n$.
	We will show that this is the case, and that we have
	\begin{equation}
		\Delta^{-n} a \Delta^n b \ket{\Omega} = a_{-n} b \ket{\Omega}.
	\end{equation}
	Since proposition \ref{prop:core} tells us that $\A_{\text{tidy}} \ket{\Omega}$ is a core for $\Delta^n,$ we can then invoke theorem \ref{thm:operator-continuation} to obtain the desired identity
	\begin{equation}
		\overline{\Delta^{-n} a \Delta^n} = a_{-n}.
	\end{equation}
	
	To proceed, note that we may write
	\begin{align}
		\begin{split}
		a \Delta^n b \ket{\Omega}
			& = a b_n \ket{\Omega} \\
			& = a b_n' \ket{\Omega} \\
			& = b_n' a \ket{\Omega} \\
			& = b_n' a' \ket{\Omega}.
		\end{split}
	\end{align}
	so $a \Delta^n b \ket{\Omega}$ is in the domain of the adjoint of the Tomita operator $S^{\dagger},$ and we have
	\begin{equation}
		S^{\dagger} a \Delta^n b \ket{\Omega}
			= (a')^{\dagger} (b_n')^{\dagger} \ket{\Omega}.
	\end{equation}
	Now applying lemma \ref{lem:dagger-ladder}, we have
	\begin{align}
	\begin{split}
		S^{\dagger} a \Delta^n b \ket{\Omega}
		& = (a')^{\dagger} (b_{n-1})^{\dagger} \ket{\Omega} \\
		& = (b_{n-1})^{\dagger} (a')^{\dagger} \ket{\Omega} \\
		& = (b_{n-1})^{\dagger} (a_{-1})^{\dagger} \ket{\Omega}.
	\end{split}
	\end{align}
	So this vector is in the domain of $S$.
	Using the identity $S S^{\dagger} = \Delta^{-1}$, it follows that $a \Delta^n b \ket{\Omega}$ is in the domain of $\Delta^{-1},$ and we have
	\begin{align}
		\begin{split}
		\Delta^{-1} a \Delta^n b \ket{\Omega}
			& = S (b_{n-1})^{\dagger} (a_{-1})^{\dagger} \ket{\Omega} \\
			& = a_{-1} b_{n-1} \ket{\Omega}.
		\end{split}
	\end{align}
	Iterating this procedure $n$ times, we see that $a \Delta^n b \ket{\Omega}$ is in the domain of $\Delta^{-n},$ and that we have
	\begin{align}
	\begin{split}
		\Delta^{-n} a \Delta^n b \ket{\Omega}
		& = a_{-n} b \ket{\Omega},
	\end{split}
	\end{align}
	as desired.
\end{proof}

\begin{corollary}
	For $a$ in the tidy subspace $\A_{\text{tidy}}$, the function
	\begin{equation}
		F_a(z) = \bar{\Delta^{-z} a \Delta^{z}}
	\end{equation}
	is norm analytic in the entire complex plane.
	For integer values of $z,$ the function $F_a$ is valued in $\A$.
\end{corollary}
\begin{proof}
	The preceding theorem tells us that at the integers, we have $F_a(n) = a_{-n},$ so $F_{a}(n)$ is in $\A$.
	
	Analyticity of $F_a(z)$ follows almost immediately from theorem \ref{thm:operator-continuation}.
	Technically that theorem only guarantees that this function is analytic in the right half-plane and the left half-plane, and strongly continuous on the imaginary axis.
	But a simple argument using Morera's theorem shows that the function is analytic on the imaginary axis as well.
	(Anyway, this isn't so important, since we will eventually be applying Carlson's theorem, and Carlson's theorem works just fine for functions analytic in a half-plane.)
\end{proof}

\begin{theorem}[Tomita's theorem on the tidy subspace] \label{thm:tidy-tomita}
	Let $a$ be an operator in the tidy subspace $\A_{\text{tidy}}$
	Then for any $t,$ the operator $\Delta^{-it} a \Delta^{it}$ is in $\A.$
\end{theorem}
\begin{proof}
	Fix $a \in \A_{\text{tidy}}.$
	It will suffice to show that for any $b' \in \A',$ the operator $\Delta^{-it} a \Delta^{it}$ commutes with $b'.$
	Since $a$ is in the tidy subspace, the previous corollary tells us that the function
	\begin{equation}
		F_a(z) = \bar{\Delta^{-z} a \Delta^z}
	\end{equation}
	is norm analytic in the entire complex plane.
	For any $b' \in \A',$ it follows that the function
	\begin{equation}
		F_{a b'}(z) = [\bar{\Delta^{-z} a \Delta^z}, b']
	\end{equation}
	is norm analytic in the entire complex plane.
	Furthermore, it vanishes on the integers.
	
	The norm of this function is bounded by
	\begin{equation}
		\lVert F_{ab'}(z) \rVert_{\infty}
			\leq 2 \lVert \bar{\Delta^{-z} a \Delta^{z}} \rVert_{\infty} \lVert b' \rVert_{\infty}.
	\end{equation}
	Since $\Delta^{it}$ is unitary for real $t,$ we may bound this by
	\begin{equation}
		\lVert F_{ab'}(z) \rVert_{\infty}
			\leq 2 \lVert \bar{\Delta^{-\Re(z)} a \Delta^{\Re(z)}} \rVert_{\infty} \lVert b' \rVert_{\infty}.
	\end{equation}
	The Phragmen-Lindelof theorem tells us that the norm of $\bar{\Delta^{-\Re(z)} a \Delta^{\Re(z)}}$ is upper bounded by the norm of $a_{-n},$ where $n$ is the nearest integer to $\Re(z)$ satisfying $|n| \geq |\Re(z)|.$
	We know by remark \ref{rem:tidy-properties} that the norms of the operators $a_n$ are exponentially bounded.
	So $F_{ab'}$ is an entire function from $\comps$ to $\B(\H)$ that vanishes on the integers, and for which the norm is bounded by an exponential function of the nearest integer bounding the real part of $z.$
	A version of Carlson's theorem appropriate for this circumstance, given explicitly in appendix \ref{app:carlson}, shows $F_{ab'}(z) = 0.$
\end{proof}

\begin{remark}
	Now we will show a sense in which $\A_{\text{tidy}}$ generates all of $\A,$ and use this to prove Tomita's theorem in generality.
\end{remark}

\begin{prop} \label{prop:density}
	We have $\A_{\text{tidy}}'' = \A.$
\end{prop}
\begin{proof}
	The inclusion $\A_{\text{tidy}} \subseteq \A$ and the double commutant theorem $\A = \A''$ give the inclusion $\A_{\text{tidy}}'' \subseteq \A.$
	We must show the reverse inclusion.
	For this it suffices to show $\A_{\text{tidy}}' \subseteq \A'.$

	To show this, recall that by proposition \ref{prop:core}, $\A_{\text{tidy}} \ket{\Omega}$ is a core for $\Delta^{1/2}.$
	It is therefore also a core for the Tomita operator $S = J \Delta^{1/2}$.
	By the properties of closed operators discussed in section \ref{sec:unbounded-operators}, it follows that for any $a \in \A,$ there exists a sequence of tidy operators $a_n$ satisfying both
	\begin{equation}
		a_n \ket{\Omega} \to a \ket{\Omega}
	\end{equation}
	and
	\begin{equation}
		a_n^{\dagger} \ket{\Omega} = S a_n \ket{\Omega} \to S a \ket{\Omega} = a^{\dagger} \ket{\Omega}.
	\end{equation}
	From this it is straightforward to show that for any operator $b'$ in $\A',$ we have the limits
	\begin{equation}
		\lim_{n \to \infty} a_n b' \Omega = a b' \ket{\Omega}
	\end{equation}
	and
	\begin{equation}
		\lim_{n \to \infty} a_n^{\dagger} b' \ket{\Omega} = a^{\dagger} b' \ket{\Omega}.
	\end{equation}
	
	Now, suppose $O$ is an operator in $\A_{\text{tidy}}',$ $a$ is an operator in $\A,$ and $a_n$ is a sequence of tidy operators converging as in the above equations.
	Fix $b', c' \in \A'.$
	We have
	\begin{align}
		\begin{split}
			\langle  c' \Omega | [O, a] | b' \Omega \rangle
			& = \langle O^{\dagger} c' \Omega | a b' \Omega \rangle
			- \langle a^{\dagger} c' \Omega  | O b' \Omega \rangle \\
			& = \lim_{n \to \infty} \left( \langle O^{\dagger} c' \Omega | a_n b' \Omega \rangle
			- \langle a_n^{\dagger} c' \Omega  | O b' \Omega \rangle\right) \\
			& = \lim_{n \to \infty} \langle c' \Omega | [O, a_n] b' \Omega \rangle \\
			& = 0.
		\end{split}
	\end{align}
	Since $\A' \ket{\Omega}$ is dense in $\H,$ this equation implies that $[O, a]$ vanishes as an  operator.
	So every element of $\A_{\text{tidy}}'$ commutes with every element of $\A,$ as desired.
\end{proof}
	
\begin{corollary}[The general version of Tomita's theorem]
	For any operator $a \in \A$ and any $t \in \mathbb{R},$ the operator $\Delta^{-it} \mathrm{a} \Delta^{it}$ is in $\A.$
	\end{corollary}
	\begin{proof}
	Let $b'$ be a tidy operator for the commutant algebra $\A',$ and fix arbitrary $\ket{\psi}, \ket{\xi} \in \H.$
	Consider the commutator
	\begin{align}
		\begin{split}
			\langle \xi | [\Delta^{-it} a \Delta^{it}, b'] | \psi \rangle
			& = \langle \xi | \Delta^{-it} a \Delta^{it} b' \psi \rangle
					- \langle \xi | b' \Delta^{-it} a \Delta^{it} \psi \rangle. 
		\end{split}
	\end{align}
	We already know via theorem \ref{thm:tidy-tomita} that Tomita's theorem holds for tidy operators.
	By applying this to $b',$ which is a tidy operator of $\A',$ we observe that the operator
	\begin{equation}
		b'(t) \equiv \Delta^{it} b' \Delta^{-it}
	\end{equation}
	is in $\A'.$
	This gives
	\begin{align}
	\begin{split}
		\langle \xi | [\Delta^{-it} a \Delta^{it}, b'] | \psi \rangle
		& = \langle \Delta^{it} \xi | a b'(t) \Delta^{it} \psi \rangle
		- \langle \Delta^{it} \xi | b'(t) a \Delta^{it} \psi \rangle \\
		& = \langle \Delta^{it} \xi | [a, b'(t)] | \Delta^{it} \psi \rangle \\
		& = 0.
	\end{split}
	\end{align}
	So $\Delta^{-it} a \Delta^{it}$ commutes with every tidy operator in $\A'.$
	By proposition \ref{prop:density}, it is in $\A.$
\end{proof}

\subsection{A note on modular conjugations}
\label{sec:modular-conjugation}

In the preceding subsection, we showed that for $a$ in the tidy subspace and for any $z \in \comps,$ the operator $\Delta^{-z} a \Delta^{z}$ is closed and densely defined, and its closure is an operator in $\A.$
Note that for $a, b, c \in \A,$ and $S$ the Tomita operator, we have
\begin{equation}
	S a S b c \ket{\Omega}
		= S a c^{\dagger} b^{\dagger} \ket{\Omega}
		= b c a^{\dagger} \ket{\Omega}
		= b S a c^{\dagger} \ket{\Omega}
		= b S a S c \ket{\Omega}.
\end{equation}
So $S a S$ commutes with $b$ when acting on $c \ket{\Omega}.$
Since vectors of the form $c \ket{\Omega}$ are dense in Hilbert space, we may conclude that $S a S$ is an operator affiliated with $\A'.$
If $a$ is in the tidy subspace, then we have
\begin{equation}
	S a S b \ket{\Omega}
		= J \Delta^{1/2} a \Delta^{-1/2} J b \ket{\Omega}.
\end{equation}
Since $\Delta^{1/2} a \Delta^{-1/2}$ is bounded on its domain, this is a bounded function of $b \ket{\Omega}.$
It therefore follows that the operator $S a S$ is bounded on its domain, so its closure is an operator in $\A'.$

We have shown that for every operator $a$ in the tidy subspace, the closure of the operator $J (\Delta^{1/2} a \Delta^{-1/2}) J$ is in $\A'.$
But since each operator $b$ in the tidy subspace can be written as the closure of $\Delta^{1/2} a \Delta^{-1/2}$ for some other operator $a = \bar{\Delta^{-1/2} b \Delta^{1/2}}$ in the tidy subspace, it follows that we have $J b J \in \A'$ for each $b$ in the tidy subspace.
 
For general $a \in \A,$ let $b$ be an operator in the tidy subspace of $\A.$
We have
\begin{align}
	\begin{split}
		[J a J, b] = J [a, J b J] J = 0,
	\end{split}
\end{align}
since $J b J$ is in $\A'.$
So $J a J$ commutes with every tidy operator in $\A.$
By proposition \ref{prop:density}, we conclude that $J a J$ is in $\A'.$
We have therefore shown the inclusion $J \A J \subseteq \A'.$
A symmetric argument gives $J \A' J \subseteq \A,$ and we conclude $J \A J = \A'.$

\acknowledgments{I thank Brent Nelson for comments on a companion paper that improved the presentation of section \ref{sec:tomitas-theorem}.
I also thank \c{S}erban Str\v{a}til\v{a} and L\'{a}szl\'{o} Zsid\'{o} for writing their wonderful book \cite{struatilua2019lectures}, where I learned most of the techniques of operator analysis that were used in this paper.
Part of this work was completed at the long term workshop YITP-T-23-01 held at YITP, Kyoto University.
Financial support was provided by the AFOSR under award number FA9550-19-1-0360, by the DOE Early Career Award, and by the Templeton Foundation via the Black Hole Initiative.}

\appendix

\section{Carlson's theorem}
\label{app:carlson}

In this appendix, I give a version of Carlson's theorem appropriate for the application in the main text.
The proof works via repeated applications of the Phragm\'{e}n-Lindel\"{o}f theorem.
I begin by stating this theorem in a form appropriate for the application to Carlson's theorem, and sketching a proof.

\begin{theorem}[Phragm\'{e}n-Lindel\"{o}f]
	Let $X$ be a Banach space, and let $S$ be a vertical strip in the complex plane of width $\ell.$
	Let $f : S \to X$ be holomorphic in the interior of the strip and continuous at its boundaries.
	Suppose that $f$ is bounded on the left and right edges of the strip, and is known to grow at most doubly exponentially in the vertical direction in the interior of the strip; specifically suppose that there exist constants $\alpha, \beta, \gamma$ with $\gamma < \pi/\ell$ and satisfying
	\begin{equation}
		\lVert f(z) \rVert \leq \alpha e^{\beta e^{\gamma |y|}}.
	\end{equation}
	Then in fact $f$ is bounded everywhere in the strip by the bound on the left and right edges.
	
	Let $W$ be an angular sector in the right half-plane, centered at the origin, and with angular extent $\theta$.
	Let $f : W \to X$ be holomorphic in the interior of the sector and continuous at its boundaries.
	Suppose that $f$ is bounded on the half-lines that bound the sector, and such that there exist constants $\alpha, \gamma$ with $\gamma < \pi/\theta$ satisfying
	\begin{equation}
		\lVert f(z) \rVert \leq \alpha e^{\beta r^{\gamma}}.
	\end{equation}
	Then in fact $f$ is bounded everywhere in $W$ by the bound on its half-line boundaries.
	In particular, if the angular extent of $W$ is less than $\pi,$ then we may take $\gamma=1.$ 
\end{theorem}
\begin{proof}[Sketch of proof]
	For the ``strip'' statement, we may take the strip to have left edge on the imaginary axis, and right edge on the vertical line of real part $\ell.$
	Fix $\epsilon > 0,$ and consider the function
	\begin{equation}
		g(z) = e^{- \epsilon \sin(\pi z/\ell)} f(z).
	\end{equation}
	This is holomorphic in the strip $S$.
	If we consider a ``box'' within the vertical strip, filling the full strip horizontally but only going from $-y_0$ to $y_0$ vertically, then because this is a compact region and $g(z)$ is holomorphic, its norm in the interior of the box is bounded above by its norm on the edge of the box.
	On the left and right edges of the box, the norm of $g(z)$ is just the norm of $f(z).$
	Because $\sin(\pi z/\ell)$ falls off as $e^{-\pi |y|/\ell}$ in vertical directions, and because of our assumptions on the vertical growth of $f,$ the magnitude of $g(z)$ on the top and bottom of the box is suppressed in the limit $y_0 \to \infty.$
	From this we may conclude that the norm of $g(z)$ within the strip $S$ is bounded by the values it takes on the left and right boundary of the strip, which is the norm of $f(z).$
	So we have
	\begin{equation}
		|e^{- \epsilon \sin(\pi z/\ell)}| \lVert f(z) \rVert \leq \max \{ \sup_y \lVert f(i y) \rVert, \sup_y \lVert f(\ell + i y) \rVert \}.
	\end{equation}
	Taking the limit $\epsilon \to 0$ completes the proof.
	
	The statement for sectors follows from the statement for strips by identifying a strip of width $\ell$ with a sector of angular extent $\ell$ using the exponential map.
\end{proof}

\begin{theorem}
	Let $f : \comps \to \B(\H)$ be a norm analytic function.
	For any  $z \in \comps,$ let $\lceil z \rceil$ denote the integer nearest to $\Re(z)$ satisfying $|n| \geq |\Re(z)|.$
	Suppose that there exist constants $\alpha, \beta > 0$ satisfying
	\begin{equation}
		\lVert f(z) \rVert_{\infty} \leq \alpha e^{\beta |\lceil z \rceil|}.
	\end{equation}
	Suppose further that $f$ vanishes on the integers.
	Then $f$ vanishes everywhere.
	
	Note that it would have been sufficient to assume $f$ was only analytic in the right half-plane and continuous on the boundary, and vanished on the nonnegative integers.
\end{theorem}
\begin{proof}
	Fix $n \geq 0,$ and consider the strip $0 \leq \Re(z) \leq n.$
	The function
	\begin{equation}
		g(z) = f(z) e^{-\beta z}
	\end{equation}
	is analytic in the strip and continuous on its boundary.
	By assumption, this function is bounded.
	So by the Phragm\'{e}n-Lindel\"{o}f theorem, the maximum of $g$ within the strip is attained on the boundary of the strip, which implies the bound
	\begin{equation}
		\lVert g(z) \rVert
			\leq \max\{ \sup_t \lVert f(0+it)\rVert,  \alpha \}.
	\end{equation}
	Since this holds for any strip, it follows that $g$ is bounded in the entire right half-plane.
	
	We now have a function $g(z) = f(z) e^{-\beta z}$ that is bounded by some constant $g_0$ in the right half-plane, analytic in that plane, and continuous on the imaginary axis.
	It also vanishes on the integers, so the function
	\begin{equation}
		h(z) = \frac{g(z)}{\sin(\pi z)}
	\end{equation}
	is analytic in the right half-plane and continuous on the imaginary axis.
	The maximum of $h(z)$ within the half-disc of radius $n+\frac{1}{2}$ is attained on its boundary, which is a half-circle of radius $n+\frac{1}{2}$ together with a segment of the imaginary axis.
	On the half-circle, the function $\frac{1}{\sin(\pi z)}$ is bounded by 1.
	So $h(z)$ is bounded, within each half-disc, by the maximum of $g_0$ and its values on the imaginary axis.
	
	Consider now the quadrant of the right half-plane with $\Im(z) \geq 0.$
	This is an angular sector of extent $\pi/2.$
	The function $e^{- i \pi z} h(z)$ is bounded on the boundaries of this sector, and at most exponentially growing in the interior.
	So by the Phragm\'{e}n-Lindel\"{o}f theorem, its maximum is attained on the boundary of the sector; it follows that $e^{- i \pi z} h(z)$ is bounded in this sector, and therefore that within this sector we have some constant $\gamma$ with $|h(z)| \leq \gamma e^{-\pi y}.$
	A similar argument for the lower sector of the right half-plane tells us that within the full right half-plane there exists a constant $\gamma$ with
	\begin{equation}
		|h(z)| \leq \gamma e^{- \pi |y|}.
	\end{equation}
	
	Finally, fix $\omega > 0,$ and consider the function
	\begin{equation}
		h_{\omega}(z) = e^{\omega z} h(z).
	\end{equation}
	We have the universal bound
	\begin{equation}
		|h_{\omega}(z)| \leq \gamma e^{\omega x - \pi |y|}.
	\end{equation}
	We have $|h_{\omega}(z)| \leq \gamma$ on the three sectors bounded by the imaginary axis and the half-lines $y = \pm \frac{\omega}{\pi} x.$
	Within each sector, $|h_{\omega}(z)|$ at most exponentially growing, so applying the Phragm\'{e}n-Lindel\"{o}f theorem gives us $|h_{\omega}(z)| \leq \gamma$ everywhere, hence $|h(z)| \leq \gamma e^{- \omega x}.$
	Taking the limit $\omega \to 0$ gives $h(z) = 0$ for $x > 0,$ and continuity gives $h(z) = 0$ everywhere in the right half-plane.
	We therefore have $g(z) = 0,$ hence $f(z) = 0.$
\end{proof}

\section{Basic properties of the Tomita operator}
\label{app:tomita}

Analysis of the KMS condition in section \ref{sec:uniqueness} motivated the introduction of an antilinear operator $S_0,$ defined on all states of the form $a \ket{\Omega}$ for $a \in \A,$ and which acts on these states as
\begin{equation}
	S_0 (a \ket{\Omega}) = a^{\dagger} \ket{\Omega}.
\end{equation}
Under the assumption that $\ket{\Omega}$ is cyclic for $\A,$ vectors of the form $a \ket{\Omega}$ form a dense subspace of Hilbert space, so the antilinear operator $S_0$ is densely defined.
We would like first to show that $S_0$ is preclosed, so that it can be closed to some better-behaved operator $S$; we will then call this operator $S$ the \textit{Tomita operator} and prove some of its important properties.
The properties we prove in this appendix are summarized in the list below.
\begin{itemize}
	\item Assume $\ket{\Omega}$ is cyclic and separating for $\A,$ and consider the densely defined operators
	\begin{align}
		S_0 (a \ket{\Omega})
			& = a^{\dagger} \ket{\Omega}, \\
		F_0 (a' \ket{\Omega})
			& = (a')^{\dagger} \ket{\Omega}.
	\end{align}
	These operators are preclosed.
	Denote their closures by $S$ and $F.$
	\item
	The domain of $S$ consists of all vectors of the form $T \ket{\Omega}$ such that either (i) $T$ is in $\A,$ or (ii) $T$ is affiliated to $\A$ with core $\A' \ket{\Omega},$ and $\ket{\Omega}$ is in the domain of both $T$ and $T^{\dagger}.$
	In either case, we have $S T \ket{\Omega} = T^{\dagger} \ket{\Omega}.$
	An analogous statement holds for the domain of $F.$
	\item We have $S^{\dagger} = F$ and $S = S^{-1},$ i.e. $S^2 = 1_{\D_S}.$
	\item Denoting the polar decomposition of $S$ by
	\begin{equation}
		S
		= J \Delta^{1/2},
	\end{equation}
	the antilinear partial isometry $J$ is called the \textbf{modular conjugation} and the operator $\Delta$ is called the \textbf{modular operator}.
	The modular operator has trivial kernel, and is therefore invertible.
	Note that this equation implies the domain of $\Delta^{1/2}$ is the same as the domain of $S$.
	\item The antilinear partial isometry $J$ is in fact an antiunitary operator, and satisfies $J^2 = 1,$ hence $J = J^{\dagger}.$
	\item The polar decomposition of $F$ is given by
	\begin{equation}
		F
		= J \Delta^{-1/2}.
	\end{equation}
	\item The operators $J$ and $\Delta$ satisfy $J \Delta^{1/2} J = \Delta^{-1/2}.$
\end{itemize}
Note that it is often cleanest to think of an antilinear operator on $\H$ as an ordinary linear operator mapping a domain in $\H$ to the complex conjugate space $\bar{\H}.$
We will not do this here, but if anything below seems confusing, the confusion can probably be resolved by thinking in these terms.

As explained in section \ref{sec:unbounded-operators}, an operator is preclosed if and only if its adjoint is densely defined.
So we can show that $S_0$ is preclosed by finding a dense subspace of $\H$ on which its adjoint is defined.
We will actually do this by showing $S_0^{\dagger} \supseteq F_0,$ so that in particular $S_0^{\dagger}$ is defined on the domain of $F_0,$ which is $\A' \ket{\Omega},$ which is dense.
A symmetric argument will imply $F_0^{\dagger} \supseteq S_0,$ so we will also know that $F_0$ is preclosed.
The properties of adjoints discussed in proposition \ref{prop:adjoint-properties} then guarantee $S \subseteq F^{\dagger},$ and we will eventually show equality.

To proceed, it will be helpful to recall the definition of the adjoint of an antilinear operator.
If $L$ is an antilinear operator with domain $\D_L$, then the adjoint $L^{\dagger}$ is defined on every vector $\ket{\xi}$ such that the map
\begin{equation}
	\ket{\psi} \mapsto \langle L \psi | \xi\rangle, \quad \ket{\psi} \in \D_L
\end{equation}
is bounded.
The action of $L^{\dagger}$ on any such vector is given via the matrix elements
\begin{equation}
	\langle \psi | L^{\dagger} \xi \rangle = \bar{\braket{L \psi}{\xi}} = \braket{\xi}{L \psi}.
\end{equation}

We will show that $S_0^{\dagger}$ can act on every vector of the form $a' \ket{\Omega}$ for $a' \in \A',$ and that it acts on those vectors in the same way as $F_0.$
For any $a' \in \A',$ and any $a \ket{\Omega}$ in the domain of $S_0,$ we have
\begin{equation}
	\braket{S_0 a \Omega}{a' \Omega}
		= \langle a^{\dagger} \Omega | a' \Omega \rangle
		= \langle (a')^{\dagger} \Omega | a \Omega \rangle.
\end{equation}
This is bounded as a function of $a \ket{\Omega},$ so $S_0^{\dagger}$ can act on $a' \ket{\Omega},$ and its action is determined by the matrix elements
\begin{equation}
	\langle a \Omega | S_0^{\dagger} a' \Omega \rangle
		= \langle a \Omega | (a')^{\dagger} \Omega \rangle.
\end{equation}
This uniquely fixes $S_0^{\dagger}$ to act (antilinearly) on $a' \ket{\Omega}$ as $S_0^{\dagger} (a' \ket{\Omega}) = (a')^{\dagger} \ket{\Omega},$ which implies $S_0^{\dagger} \supseteq F_0.$

We now know that $S_0$ and $F_0$ are preclosed.
We denote their closures by $S$ and $F$.
Since $S_0^{\dagger}$ is closed (cf. section \ref{sec:unbounded-operators}), we have $F \subseteq S^{\dagger},$ and by proposition \ref{prop:adjoint-properties}, we also have $S \subseteq F^{\dagger}.$
To show equality, we will need to study the domain of $S$.
This leads us to the following theorem.
\begin{theorem}[Domain of the Tomita operator]
	If a vector $\ket{\psi}$ is in the domain of $S,$ then there exists a closed operator $T$ affiliated with $\A$, such that $\ket{\Omega}$ is in the domain of both $T$ and $T^{\dagger},$ and satisfying
	\begin{align}
		\ket{\psi}
			& = T \ket{\Omega}, \\
		S \ket{\psi}
			& = T^{\dagger} \ket{\Omega}.
	\end{align}
	Note also that because $T$ is affiliated to $\A,$ it is defined on $\A' \ket{\Omega},$ and can be uniquely determined if one requires that $\A' \ket{\Omega}$ is a core for $T$.
	
	Conversely, if $T$ is a closed operator affiliated to $\A$ with $\ket{\Omega}$ in the domain of both $T$ and $T^{\dagger}$, then $T \ket{\Omega}$ is in the domain of $S$ and satisfies $S T \ket{\Omega} = T^{\dagger} \ket{\Omega}$.
\end{theorem}
\begin{proof}
	We proceed by defining operators that act on $\A' \ket{\Omega}$ the way we think $T$ and $T^{\dagger}$ should act, supposing they exist.
	Since $T$ and $T^{\dagger}$ are supposed to be affiliated with $\A,$ they should act on vectors of the form $a' \ket{\Omega}$ as
	\begin{equation}
		T a' \ket{\Omega}
			= a' T \ket{\Omega}
			= a' \ket{\psi}
	\end{equation}
	and
	\begin{equation}
		T^{\dagger} a' \ket{\Omega}
			= a' T^{\dagger} \ket{\Omega}
			= a' S \ket{\psi}.
	\end{equation}
	So we simply define operators $\alpha_0$ and $\beta_0$ on the domain $\A' \ket{\Omega}$ by
	\begin{align}
		\alpha_0 a' \ket{\Omega}
			& = a' \ket{\psi}, \\
		\beta_0 a' \ket{\Omega}
			& = a' S \ket{\psi}.
	\end{align}
	For $b' \in \A',$ we have
	\begin{align}
		\begin{split}
		\langle b' \Omega | \alpha_0 | a' \Omega \rangle
			& = \langle b' \Omega | a' \psi \rangle \\
			& = \langle (a')^{\dagger} b' \Omega | \psi \rangle \\
			& = \langle F (b')^{\dagger} a' \Omega | \psi \rangle \\
			& = \langle F^{\dagger} \psi | (b')^{\dagger} a' \Omega \rangle.
		\end{split}
	\end{align}
	Since we know $F^{\dagger} \supseteq S,$ and $\ket{\psi}$ is in the domain of $S$, this implies
	\begin{align}
	\begin{split}
		\langle b' \Omega | \alpha_0 | a' \Omega \rangle
			& = \langle S \psi | (b')^{\dagger} a' \Omega \rangle \\
			& = \langle b' S \psi | a' \Omega \rangle \\
			& = \langle \beta_0 b' \Omega | a' \Omega \rangle.
	\end{split}
	\end{align}
	This gives us the inclusion $\alpha_0^{\dagger} \supseteq \beta_0,$ and a symmetric argument gives $\beta_0^{\dagger} \supseteq \alpha_0.$
	From this we conclude that $\alpha_0$ is preclosed, and we denote by $T$ its closure.
	To see that $T$ is affiliated with $\A,$ note that we have
	\begin{equation}
		T a' b' \ket{\Omega}
			= \alpha_0 a' b' \ket{\Omega}
			= a' b' \ket{\psi}
			= a' \alpha_0 b' \ket{\Omega}
			= a' T b' \ket{\Omega}.
	\end{equation}
	So $T$ commutes with each $a' \in \A'$ on its core, and therefore commutes with each $a' \in \A'$ whenever the products $T a'$ and $a' T$ are both defined.
	The vector $\ket{\Omega}$ is in $\A' \ket{\Omega},$ so it is in the domain of both $\alpha_0$ and $\beta_0$, and therefore in the domain of both $T$ and $T^{\dagger}.$
	We have
	\begin{equation}
		S T \ket{\Omega}
			= S \alpha_0 \ket{\Omega} = S \ket{\psi} = \beta_0 \ket{\Omega} = T^{\dagger} \ket{\Omega}.
	\end{equation}

	For the converse, note that because $T$ is closed, it has a polar decomposition $T = u |T|$ (\ref{thm:polar-decomposition})
	Consider the operators $T_n = u |T|_n$ defined by projecting $|T|$ onto the spectral subspace $[0, n].$
	By theorem \ref{thm:neumann-polar}, each $T_n$ is in $\A,$ so each $T_n \ket{\Omega}$ is in the domain of $S$ and maps to $T_n^{\dagger} \ket{\Omega} = u^{\dagger} |T^{\dagger}|_n \ket{\Omega}.$
	The spectral theorem guarantees that these sequences converge:
	\begin{align}
		T_n \ket{\Omega}
			& \to T \ket{\Omega} \\
		S T_n \ket{\Omega}
			& \to T^{\dagger} \ket{\Omega}.
	\end{align}
	Since $S$ is a closed operator, this implies that $T \ket{\Omega}$ is in the domain of $S$, with $S T \ket{\Omega} = T^{\dagger} \ket{\Omega}.$
\end{proof}

Now we wish to use this theorem to upgrade the inclusion $S \subseteq F^{\dagger}$ to an equality.
Suppose a vector $\ket{\psi}$ is in the domain of $F^{\dagger}.$
We want to show that $\ket{\psi}$ is in the domain of $S$, which means we want to show there exists an affiliated operator $T$ like the one in the preceding theorem, with $\ket{\psi} = T \ket{\Omega}.$
It is straightforward to verify that in the proof of the preceding theorem, the operator $T$ could have been constructed only under the assumption $\ket{\psi} \in \D_{F^{\dagger}},$ without the seemingly more stringent assumption $\ket{\psi} \in \D_S.$
But the ``converse'' part of the above theorem then guarantees that we do in fact have $\ket{\psi} \in \D_S.$
We therefore have $S = F^{\dagger},$ as desired.
The above characterization of the domain of $S$ also clearly implies $S^2 = 1_{\D_S}.$

It is a general fact --- see section \ref{sec:unbounded-operators} --- that any closed operator admits a unique polar decomposition.
We write the polar decomposition of $S$ as
\begin{equation}
	S = J \Delta^{1/2},
\end{equation}
where $\Delta^{1/2}$ is a positive, self-adjoint operator given by $\Delta^{1/2} = \sqrt{S^{\dagger} S},$ and $J$ is an antilinear partial isometry whose support is given by $\text{supp}(J) = \ker(S)^{\perp}$.
The identity $S^2 = 1_{\D_S}$ implies that the kernel of $S$ is trivial, so $J$ is supported on all of Hilbert space and is therefore an antiunitary operator.
We also have that the kernel of $\Delta$ is the same as the kernel of $S$, so $\Delta$ is invertible on its domain.

We would now like to show the identities $J^2 = 1$ and $J\Delta^{1/2} J=\Delta^{-1/2}.$
We first observe that thanks to the identity $S^2 = 1_{\D_S},$ we have
\begin{equation}
	1_{\D_S} = J \Delta^{1/2} J \Delta^{1/2},
\end{equation}
which gives
\begin{equation} \label{eq:JDJ-inclusion}
	J \Delta^{1/2} J \supseteq \Delta^{-1/2}.
\end{equation}
To show equality, consider an arbitrary $\ket{\psi}$ in the domain of $J \Delta^{1/2} J.$
The vector $J \ket{\psi}$ is in the domain of $\Delta^{1/2},$ which is equal to the domain of $S.$
The identity $1_{\D_S} = S^2$ tells us that the domain of $S$ is equal to the image of $S$.
So there exists some $\ket{\eta}$ in the domain of $S$ with
\begin{equation}
	J \ket{\psi} = S \ket{\eta} = J \Delta^{1/2} \ket{\eta}.
\end{equation}
Since $J$ is antiunitary, we may left multiply by its inverse to obtain
\begin{equation}
	\ket{\psi} = \Delta^{1/2} \ket{\eta}.
\end{equation}
It follows that every vector in the domain of $J \Delta^{1/2} J$ is in the image of $\Delta^{1/2},$ hence in the domain of $\Delta^{-1/2},$ so that the inclusion in expression \eqref{eq:JDJ-inclusion} is an equality. 
To see $J^2 = 1,$ we simply observe
\begin{equation}
	\Delta^{-1/2} = J \Delta^{1/2} J = J (J J^{\dagger}) \Delta^{1/2} J = J^2 (J^{\dagger} \Delta^{1/2} J).
\end{equation}
The operator $J^{\dagger} \Delta^{1/2} J$ is positive, and the operator $J^2$ is unitary, so by uniqueness of the polar decomposition of the closed operator $\Delta^{-1/2}$, we must have $J^2 = 1.$

The last identity we want to show, $F = J \Delta^{-1/2},$ follows from what we have already shown via
\begin{equation}
	F = S^{\dagger} = \Delta^{1/2} J^{\dagger} = \Delta^{1/2} J = J \Delta^{-1/2}. 
\end{equation}

\section{van Daele's proof}
\label{app:van-daele}

This appendix gives a telegraphic outline of the structure of van Daele's proof of Tomita's theorem \cite{van-daele-proof}.
This proof is very concise, but I find it unintuitive, which is why I presented a new proof in the main text of this paper.
What follows is my best attempt to describe the techniques applied in van Daele's proof.

The idea is that for $a \in \A,$ instead of showing that the modular flow $\Delta^{-it} a \Delta^{it}$ is in $\A,$ one considers the operator-valued function $t \mapsto \Delta^{-it} a \Delta^{it}$, and tries to show that its Fourier transform
\begin{equation} \label{eq:naive-fourier-transform}
	p \mapsto \int dt\, e^{i p t} \Delta^{-it} a \Delta^{it}
\end{equation}
is in $\A.$
One then tries to show that the Fourier inversion formula converges appropriately to guarantee $\Delta^{-it} a \Delta^{it} \in \A.$

One problem is the integral in equation \eqref{eq:naive-fourier-transform} does not actually converge.
But if we multiply $\Delta^{-it} a \Delta^{it}$ by a mollifying function $f(t)$ that dies sufficiently quickly at infinity, then the Fourier transform
\begin{equation}
	p \mapsto \int dt\, e^{i p t} f(t) \Delta^{-it} a \Delta^{it}
\end{equation}
converges.
If $f(t)$ is nonzero, then the question of whether $\Delta^{-it} a \Delta^{it}$ is in $\A$ is exactly the same as the question of whether $f(t) \Delta^{-it} a \Delta^{it}$ is in $\A,$ so the presence of the mollifying function does not affect the logic of the proof.
A convenient choice of mollifying function is
\begin{equation}
	f(t) = \frac{1}{\cosh{\pi t}},
\end{equation}
because this is periodic in the real direction of the complex plane, which can be used to evaluate the integral appearing in the Fourier transform by ``closing a contour'' and ``picking up a pole.''
Concretely, the integral
\begin{equation}
	\int_{i t} dz\, e^{z p} \frac{\Delta^{-z} a \Delta^{z}}{\cos(\pi z)} - \int_{-1 - i t} dz\, e^{z p} \frac{\Delta^{-z} a \Delta^{z}}{\cos{\pi z}}
\end{equation}
should in some sense be given by a residue at $z=-1/2,$ which we expect to look like
\begin{equation}
	2 i e^{-p/2} \Delta^{-1/2} a \Delta^{1/2}.
\end{equation}
These are integrals of unbounded operators, so one has to be careful about manipulating them, but by naively parametrizing the contour integral, one expects an expression of the following kind to hold:
\begin{equation}
		Q(p)
			+ e^{-p} \Delta^{-1} Q(p) \Delta
			= 2 e^{-p/2} \Delta^{-1/2} a \Delta^{1/2},
\end{equation}
with
\begin{equation}
	Q(p) = \int dt\, e^{i t p} \frac{\Delta^{- i t} a \Delta^{i t}}{\cosh(\pi t)}.
\end{equation}
Or perhaps, multiplying on the left by $e^{p/2} \Delta^{1/2}$ and on the right by $\Delta^{-1/2},$ one expects to have an expression like
\begin{equation}
	e^{p/2 }\Delta^{1/2} Q(p) \Delta^{-1/2}
	+ e^{-p/2} \Delta^{-1/2} Q(p) \Delta^{1/2}
	= 2 a.
\end{equation}

By being careful, one can show that the \textit{actual} sense in which this expression is to hold is as an expression for matrix elements of vectors in an appropriate domain.
In particular, for any vectors $\ket{\psi}$ and $\ket{\xi}$ that are in the domain of both $\Delta^{1/2}$ and $\Delta^{-1/2},$ one can show the identity
\begin{equation}
	e^{p/2 } \langle \Delta^{1/2} \psi | Q(p) | \Delta^{-1/2} \xi \rangle 
		+ e^{-p/2} \langle \Delta^{-1/2} \psi | Q(p) | \Delta^{1/2} \xi \rangle
		= 2 \langle \psi|a|\xi\rangle.
\end{equation}
From here, the basic idea of the proof is to show that this expression can be ``solved'' for $Q(p)$ in that there is only one operator $Q(p)$ satisfying this equation for any given $p,$ and to explicitly show that $Q(p)$ can be written as an operator in $\A$ obtained via a clever application of Takesaki's resolvent lemma (cf. section \ref{sec:resolvent-lemma}).

In a bit more detail, one actually considers the operator $\tilde{Q}(p) = J Q(p) J,$ with $J$ the modular conjugation, and studies the expression
\begin{equation}
	e^{p/2 } \langle \Delta^{1/2} \psi | J \tilde{Q}(p) J | \Delta^{-1/2} \xi \rangle 
	+ e^{-p/2} \langle \Delta^{-1/2} \psi | J \tilde{Q}(p) J | \Delta^{1/2} \xi \rangle
	= 2 \langle \psi|a|\xi\rangle.
\end{equation}
One then shows, using the expression $S = J \Delta^{1/2}$ for the Tomita operator, that this equation is solved by an operator $(a')^{\dagger},$ where $a'$ appears in a resolvent-lemma expression like
\begin{equation}
	a' \ket{\Omega} = c (\lambda - \Delta^{-1})^{-1} a \ket{\Omega}
\end{equation}
for a $p$-independent real number $c$ and a complex number $\lambda$ that is some concrete function of $p.$
To actually show this, one must approximate the vectors $\ket{\psi}$ and $\ket{\xi}$ by a sequence of vectors of the form $(1 + \Delta^{-1}) b \ket{\Omega}$ for $b \in \A,$ since these are the vectors for which van Daele's equation is under algebraic control.
Once all this is done successfully (and once one also shows that van Daele's equation has a unique solution), it may be concluded that $\tilde{Q}(p)$ is equal to $(a')^{\dagger}.$
This tells us that the Fourier integral
\begin{equation}
	\int dt\, e^{i p t} \frac{J \Delta^{-it} a \Delta^{it} J}{\cosh{\pi t}}
\end{equation}
is in $\A',$ and gives an expression satisfied by this operator in terms of the resolvent of $\Delta^{-1}$ acting on the vector $a \ket{\Omega}.$
By being sufficiently careful about what it means to invert a Fourier transform of an operator-valued function, one can use this to show that $J \Delta^{-it} a \Delta^{it} J$ is in $\A'$ for any $t,$ so in particular for $t = 0$ we have $J a J \in \A'.$
A symmetric argument shows $J \A' J = \A,$ which gives
\begin{equation}
	\Delta^{-it} a \Delta^{it}
		= J (J \Delta^{-it} a \Delta^{it} J) J
		\in J \A' J
		= \A.
\end{equation}

\bibliographystyle{JHEP}
\bibliography{bibliography}

\end{document}